\documentclass[12pt]{article}
\usepackage{amsmath}
\usepackage{graphicx}
\usepackage{enumerate}
\usepackage{natbib} 
\usepackage{url} 
\usepackage{pdflscape}
\usepackage{setspace}
\usepackage{multirow}
\usepackage{hyperref}
\usepackage{amsthm}
\usepackage{amsfonts}
\usepackage{bbm}
\RequirePackage[dvipsnames]{xcolor}
\usepackage{mathptmx}
\usepackage{breqn}
\usepackage{booktabs}
\usepackage{authblk}
\usepackage[ruled,vlined]{algorithm2e}

\numberwithin{equation}{section}

\newcommand{\NN}{\mathbb{N}}

\newcommand{\RR}{\mathbb{R}}

\DeclareMathOperator*{\argmax}{arg\,max}
\DeclareMathOperator*{\argmin}{arg\,min}

\newcommand{\RC}[2]{\mathrm{RC}_{#2}(#1)}

\newtheorem{proposition}{Proposition}[section]

\newtheorem{corollary}{Corollary}[section]

\theoremstyle{remark}
\newtheorem*{remark}{Remark}
\theoremstyle{definition}

\begin{document}

\def\spacingset#1{\renewcommand{\baselinestretch}%
{#1}\small\normalsize} \spacingset{1}


\title{Improving estimation for asymptotically independent bivariate extremes via global estimators for the angular dependence function}
\author[1,2*]{C. J. R. Murphy-Barltrop}
\author[3]{J. L. Wadsworth}
\author[3]{E. F. Eastoe}
\affil[1]{Technische Universität Dresden, Institut Für Mathematische Stochastik, Helmholtzstraße 10, 01069 Dresden, Germany}
\affil[2]{Center for Scalable Data Analytics and Artificial Intelligence (ScaDS.AI) Dresden/Leipzig, Germany}
\affil[3]{Department of Mathematics and Statistics, Lancaster University LA1 4YF, United Kingdom}
\affil[*]{Correspondence to: callum.murphy-barltrop@tu-dresden.de}
\date{\today}

\maketitle

\bigskip
\begin{abstract}
    Modelling the extremal dependence of bivariate variables is important in a wide variety of practical applications, including environmental planning, catastrophe modelling and hydrology. The majority of these approaches are based on the framework of bivariate regular variation, and a wide range of literature is available for estimating the dependence structure in this setting. However, such procedures are only applicable to variables exhibiting asymptotic dependence, even though asymptotic independence is often observed in practice. In this paper, we consider the so-called `angular dependence function'; this quantity summarises the extremal dependence structure for asymptotically independent variables. Until recently, only pointwise estimators of the angular dependence function have been available. We introduce a range of global estimators and compare them to another recently introduced technique for global estimation through a systematic simulation study, and a case study on river flow data from the north of England, UK. 
\end{abstract}

\noindent%
{\it Keywords:}  Bivariate Extremes, Dependence Modelling, Asymptotic Independence, Angular Dependence Function
\vfill

\newpage
\spacingset{1.8}

\section{Introduction} \label{Sec1}
Bivariate extreme value theory is a branch of statistics that deals with the modelling of dependence between the extremes of two variables. This type of analysis is useful in a variety of fields, including finance \citep{Castro-Camilo2018}, engineering \citep{Ross2020}, and environmental science \citep{Brunner2016}, where understanding and predicting the behaviour of rare, high-impact events is important. 

In certain applications, interest lies in understanding the risk of observing simultaneous extreme events at multiple locations; for example, in the context of flood risk modelling, widespread flooding can result in damaging consequences to properties, businesses, infrastructure, communications and the economy \citep{Lamb2010,Keef2013a}. To support resilience planning, it it imperative to identify locations at high risk of joint extremes. 

Classical theory for bivariate extremes is based on the framework of regular variation. Given a random vector $(X,Y)$ with standard exponential margins, we say that $(X,Y)$ is bivariate regularly varying if, for any measurable $B \subset [0,1]$,
\begin{equation} \label{eqn:MRV} 
    \lim_{r \to \infty} \Pr (V \in B, R > sr \mid R>r) = H(B) s^{-1}, \; s \geq 1,
\end{equation}
with $R:= e^X+e^Y$, $V := e^X/R$  and $H(\partial B) = 0$, where $\partial B$ is the boundary of $B$ \citep{Resnick1987}. Note that bivariate regular variation is most naturally expressed on standard Pareto margins, and the mapping $(X,Y) \mapsto (e^X,e^Y)$ performs this transformation. We refer to $R$ and $V$ as radial and angular components, respectively. Equation \eqref{eqn:MRV} implies that for the largest radial values, the radial and angular components are independent. Furthermore, the quantity $H$, which is known as the spectral measure, must satisfy the moment constraint $\int_0^1v\mathrm{d}H(v) = 1/2$. 

The spectral measure summarises the extremal dependence of $(X,Y)$, and a wide range of approaches exist for its estimation \citep[e.g.,][]{Einmahl2009,Carvalho2014,Eastoe2014}. Equivalently, one can consider Pickands' dependence function \citep{Pickands1981}, which has a direct relationship to $H$ via
\begin{equation*}
    A(t) = \int_0^1\max\{vt, (1-v)(1-t) \}2\mathrm{d}H(v), \hspace{1em} t \in [0,1],
\end{equation*}
where $A$ is a convex function satisfying $\max(t,1-t) \leq A(t) \leq 1$. This function again captures the extremal dependence of $(X,Y)$, and many approaches also exist for its estimation \citep[e.g.,][]{Guillotte2016,Marcon2016,Vettori2018}. Moreover, estimation procedures for the spectral measure and Pickands' dependence function encompass a wide range of statistical methodologies, with parametric, semi-parametric, and non-parametric modelling techniques proposed in both Bayesian and frequentist settings. 

However, methods based on bivariate regular variation are limited in the forms of extremal dependence they can capture. This dependence can be classified through the coefficient $\chi$ \citep{Joe1997}, defined as
\begin{equation*}
    \chi := \lim_{u \to \infty}\Pr(Y > u \mid X > u) \in [0,1],
\end{equation*}
where this limit exists. If $\chi >0$, then $X$ and $Y$ are asymptotically dependent, and the most extreme values of either variable can occur simultaneously. If $\chi = 0$, $X$ and $Y$ are asymptotically independent, and the most extreme values of either variable occur separately. 

Under asymptotic independence, the spectral measure $H$ places all mass on the points $\{0\}$ and $\{1\}$; equivalently, $A(t) = 1$ for all $t \in [0,1]$. Consequently, for this form of dependence, the framework given in equation \eqref{eqn:MRV} is degenerate and is unable to accurately extrapolate into the joint tail \citep{Ledford1996,Ledford1997}. Practically, an incorrect assumption of asymptotic dependence between two variables is likely to result in an overly conservative estimate of joint risk.

To overcome this limitation, several models have been proposed that can capture both classes of extremal dependence. The first was given by \citet{Ledford1996}, in which they assume that as $u \to \infty$, the joint tail can be represented as 
\begin{equation} \label{eqn:led_tawn}
    \Pr (X > u, Y > u) = \Pr\{\min(X,Y)>u\} = L(e^u)e^{-u/\eta},
\end{equation}
where $L$ is a slowly varying function at infinity, i.e., $\lim_{u \to \infty}L(cu)/L(u) = 1$ for $c>0$, and $\eta \in (0,1]$. The quantity $\eta$ is termed the coefficient of tail dependence, with $\eta=1$ and $\lim_{u \to \infty}L(u) > 0$ corresponding to asymptotic dependence and either $\eta < 1$ or $\eta = 1$ and $\lim_{u\to \infty}L(u) = 0$ corresponding to asymptotic independence. Many extensions to this approach exist \citep[e.g.,][]{Ledford1997,Resnick2002,Ramos2009}; however, all such approaches are only applicable in regions where both variables are large, limiting their use in many practical settings. Since many extremal bivariate risk measures, such as environmental contours \citep{Haselsteiner2021} and return curves \citep{Murphy-Barltrop2023}, are defined both in regions where both variables are extreme and in regions where only one variable is extreme, methods based on equation \eqref{eqn:led_tawn} are inadequate for their estimation.   

Several copula-based models have been proposed that can capture both classes of extremal dependence, such as those given in \citet{Coles2002}, \citet{Wadsworth2017} and \citet{Huser2019}. Unlike equation \eqref{eqn:led_tawn}, these can be used to evaluate joint tail behaviour in all regions where at least one variable is extreme. However, these techniques typically require strong assumptions about the parametric form of the bivariate distribution, thereby offering reduced flexibility.

\citet{Heffernan2004} proposed a modelling approach, known as the conditional extremes model, which also overcomes the limitations of the framework described in equation \eqref{eqn:led_tawn}. This approach assumes the existence of normalising functions $a:\RR_+ \to \RR$ and $b:\RR_+ \to \RR_+$ such that
\begin{equation} \label{eqn:heff_tawn}
    \lim_{u \to \infty}\Pr\left[\{Y-a(X)\}/b(X) \leq z, \; X-u > x \mid X>u\right] = D(z)e^{-x}, \; x>0, 
\end{equation}
where $D$ is a non-degenerate distribution function that places no mass at infinity. Note that the choice of conditioning on $X>u$ is arbitrary, and an equivalent formulation exists for normalised $X$ given $Y>u$. This framework can capture both asymptotic dependence and asymptotic independence, with the former arising when $a(x) = x$ and $b(x) = 1$, and can also be used to describe extremal behaviour in regions where only one variable is large. 

Finally, \citet{Wadsworth2013} proposed a general extension of equation \eqref{eqn:led_tawn}. As $u \to \infty$, they assume that for any $(\beta,\gamma) \in \RR_+^2 \setminus \{ \boldsymbol{0}\}$,
\begin{equation} \label{eqn:wads_tawn_kappa}
    \Pr(X > \beta u,Y> \gamma u) = L(e^u;\beta,\gamma)e^{-\kappa(\beta,\gamma)u},
\end{equation}
where $L(\cdot \; ;\gamma,\beta)$ is slowly varying and the function $\kappa$ provides information about the joint tail behaviour of $(X,Y)$. One can observe that equation \eqref{eqn:led_tawn} is a special case of \eqref{eqn:wads_tawn_kappa} with $\beta = \gamma$. The dependence function $\kappa$ satisfies several theoretical properties: for instance, it is non-decreasing in each argument, satisfies the lower bound $\kappa(\beta,\gamma) \geq \max\{\beta,\gamma\}$, and is homogeneuous of order 1, i.e., $\kappa(h\beta,h\gamma) = h\kappa(\beta,\gamma)$ for any $h > 0$. Setting $w := \beta/(\beta + \gamma) \in [0,1]$, the latter property implies that $\kappa(\beta,\gamma) = (\beta + \gamma)\kappa(w,1-w)$, motivating the definition of the so-called angular dependence function (ADF) $\lambda(w) = \kappa(w,1-w), \; w \in [0,1]$. Using this representation, equation \eqref{eqn:wads_tawn_kappa} can be rewritten as 
\begin{equation} \label{eqn:wads_tawn}
    \Pr(\min\{X/w,Y/(1-w)\}>u) = L(e^u;w)e^{-\lambda(w)u},  \; \; w \in [0,1], \; \lambda(w) \geq \max(w,1-w),
\end{equation}
as $u \to \infty$, where $L(\cdot \; ;w)$ is slowly varying. The ADF generalises the coefficient $\eta$, with $\eta = 1/\{2\lambda(0.5)\}$. This extension captures both extremal dependence regimes, with asymptotic dependence implying the lower bound, i.e., $\lambda(w) = \max(w,1-w)$ for all $w \in [0,1]$. Evaluation of the ADF for rays $w$ close to $0$ and $1$ corresponds to regions where one variable is larger than the other.


The ADF can be viewed as the counterpart of the Pickands' dependence function for asymptotically independent variables, and shares many of its theoretical properties \citep{Wadsworth2013}. Specifically, $\lambda(0) = \lambda(1) = 1$ and $\max(w,1-w) \leq \lambda(w)$, although there is no requirement for $\lambda(w)$ to be convex, or that it is bounded above, unlike $A(t)$. There do, however, exist shape constraints that $\lambda$ must satisfy; see Section \ref{sec:theory} for further details. The ADF can be used to differentiate between different forms of asymptotic independence, with both positive and negative associations captured, alongside complete independence, which implies $\lambda(w) = 1$ for all $w \in [0,1]$. Figure \ref{fig:Sec1_lambdas} illustrates the ADFs for three copulas. We observe a variety in shapes, corresponding to differing degrees of positive extremal dependence in the underlying copulas. The weakest dependence is observed for the inverted logistic copula, while the ADF for the asymptotically dependent logistic copula is equal to the lower bound. 

\begin{figure}[!h]
    \centering
    \includegraphics[width=\textwidth]{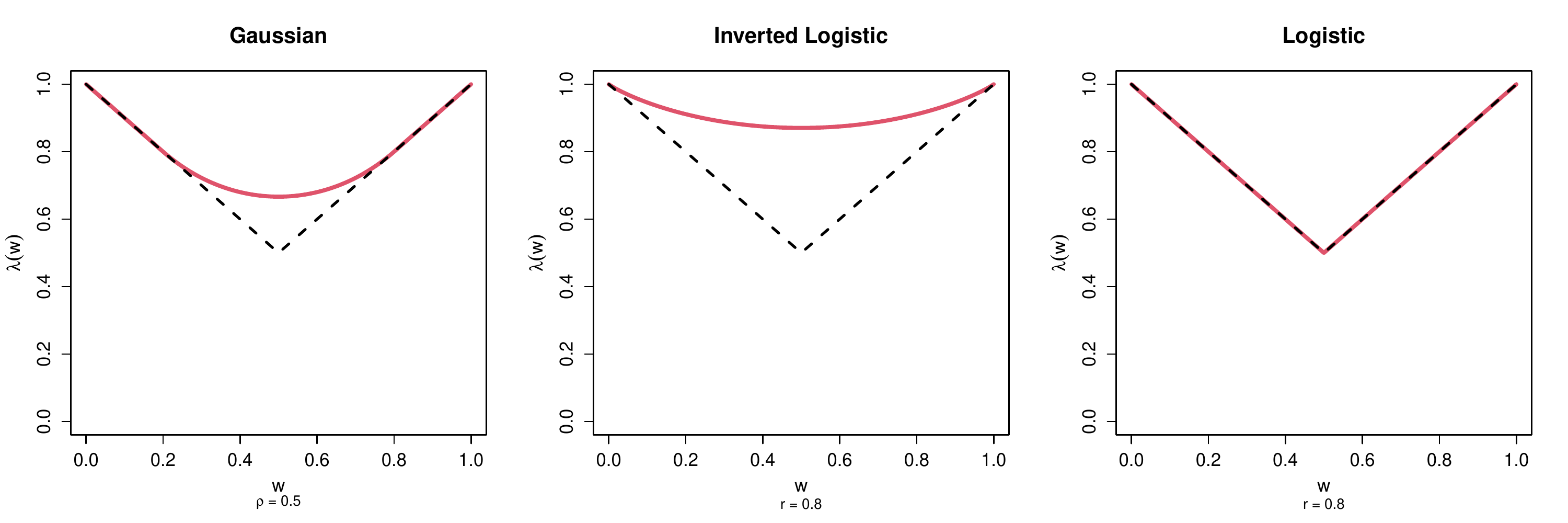}
    \caption{The true ADFs (given in red) for three example copulas. Left: bivariate Gaussian copula with coefficient $\rho = 0.5$. Centre: inverted logistic copula with dependence parameter $r = 0.8$. Right: logistic copula with dependence parameter $r = 0.8$. The lower bound for the ADF is denoted by the black dotted line. }
    \label{fig:Sec1_lambdas}
\end{figure}

Despite these modelling advances, the majority of approaches for quantifying the risk of bivariate extreme events still require bivariate regular variation. Many of the 
procedures that do allow for asymptotic independence use the conditional extremes model of equation \eqref{eqn:heff_tawn} despite some well known limitations of this approach \citep{Liu2014}. 

One particular application of the model described in equation \eqref{eqn:wads_tawn} is the estimation of so-called bivariate return curves, $\RC{p}{} := \left\{ (x,y) \in \RR^2: \Pr(X>x,Y>y) = p \right\}$, which requires knowledge of extremal dependence in regions where either variable is large; see Section \ref{Subsec5.4}. \citet{Murphy-Barltrop2023} obtain estimates of return curves, finding that estimates derived using equation \eqref{eqn:wads_tawn} were preferable to those from the conditional extremes model. \citet{Mhalla2019} and \citet{Murphy-Barltrop2022} also provide non-stationary extensions and inference methods for the ADF.

In this paper, we propose a global methodology for ADF estimation in order to improve extrapolation into the joint upper tail for bivariate random vectors exhibiting asymptotic independence. Until recently, the ADF has been estimated only in a pointwise manner using the Hill estimator \citep{Hill1975} on the tail of $\min\{X/w,Y/(1-w)\}$, resulting in unrealistic rough functional estimates and, as we demonstrate in Section \ref{Sec4}, high degrees of variability. Further, \citet{Murphy-Barltrop2023} showed that pointwise ADF estimates result in non-smooth return curve estimates, which are again unrealistic.

The first smooth ADF estimator was proposed recently in \citet{Simpson2022} based on a theoretical link between a limit set derived from the shape of appropriately scaled sample clouds and the ADF \citep{Nolde2022}. The authors introduce global estimation techniques for the limit set, from which smooth ADF estimates follow; see Section \ref{Sec2} for further details. 

We introduce several novel smooth ADF estimators, and compare their performance with the pointwise Hill estimator, as well as the estimator given in \citet{Simpson2022}. In Section \ref{Sec2}, we review the literature on ADF estimation. In Section \ref{sec:theory}, we introduce new theoretical results that the ADF must satisfy to be valid. In Section \ref{Sec3}, we introduce a range of novel estimators, and select tuning parameters for each proposed estimation technique. In Section \ref{Sec4}, we compare each of the available estimators through a systematic simulation study, finding certain estimators to be favourable over others. A subset of estimators are then applied to river flow data sets in Section \ref{Sec5} and used to obtain estimates of return curves for different combinations of river gauges. We conclude in Section \ref{Sec6} with a discussion. 

\section{Existing techniques for ADF estimation} \label{Sec2}

In this section, we introduce existing estimators for the ADF, with $(X,Y)$ denoting a random vector with standard exponential margins throughout. To begin, for any ray $w \in [0,1]$, define the min-projection at $w$ as $T_w:= \min\{X/w,Y/(1-w)\}$. Equation \eqref{eqn:wads_tawn} implies that for any $w \in [0,1]$ and $t > 0$, 
\begin{equation} \label{eqn:cond_wads_tawn}
    \Pr( T_w > u + t \mid T_w > u) = \frac{L(e^{u+t};w)}{L(e^u;w)} e^{-\lambda(w)t} \to e^{-\lambda(w)t} = t_{*}^{-\lambda(w)}, 
\end{equation}
as $u \to \infty$, with $t_{*} := e^t$. Since the expression in equation \eqref{eqn:cond_wads_tawn} has a univariate regularly varying tail with positive index, \citet{Wadsworth2013} propose using the Hill estimator \citep{Hill1975} to obtain a pointwise estimator of the ADF; we denote this `base' estimator $\hat{\lambda}_{H}$. A major drawback of this technique is that the estimator is pointwise, that is, $\lambda(w)$ is estimated separately for each $w$, leading to rough and often unrealistic estimates of the ADF. In particular, no information is shared across different rays, increasing the variability in the resulting estimates.  Furthermore, this estimator need not satisfy the theoretical constraints on the ADF identified in \citet{Wadsworth2013}, such as the endpoint conditions $\lambda(0) =\lambda(1) = 1$. 

\citet{Simpson2022} recently proposed a novel estimator for the ADF using a theoretical link with the limiting shape of scaled sample clouds. Let $C_n := \{ (X_i,Y_i)/\log n; \; i = 1, \hdots, n\}$ denote $n$ scaled, independent copies of $(X,Y)$. \citet{Nolde2022} explain how, as $n \to \infty$, the asymptotic shape of $C_n$ provides information on the underlying extremal dependence structure. In many situations, $C_n$ converges onto the compact limit set $G = \{(x,y): g(x,y) \leq 1\} \subseteq [0,1]^2$, where $g$ is the gauge function of $G$. A sufficient condition for this convergence to occur is that the joint density, $f$, of $(X,Y)$ exists, and that 
\begin{equation} \label{eqn:gauge}
    -\log f(tx,ty) \sim tg(x,y), \; \; x,y \geq 0, \; t \to \infty, 
\end{equation}
for continuous $g$. Following \citet{Nolde2014}, we also define the unit-level, boundary set $\partial G = \{(x,y): g(x,y) = 1\} \subset [0,1]^2.$ Given fixed margins, the shapes of $G$, and hence $\partial G$, are completely determined by the extremal dependence structure of $(X,Y)$. Furthermore, \citet{Nolde2022} show that the shape of $G$ is also directly linked to the modelling frameworks described in equations \eqref{eqn:led_tawn}, \eqref{eqn:heff_tawn} and \eqref{eqn:wads_tawn}, as well as the approach of \citet{Simpson2020}. In particular, letting $R_w := (w/\max(w,1-w),\infty] \times ((1-w)/\max(w,1-w),\infty]$ for all $w \in [0,1]$, we have that 
\begin{equation} \label{eqn:gauge_lambda}
    \lambda(w) = \max(w,1-w) \times r_w^{-1},
\end{equation}
where 
\begin{equation*}
    r_w = \min \{ r \in [0,1] : rR_w \cap G = \emptyset \}.
\end{equation*}
The boundary sets $\partial G$ for each of the copulas in Figure \ref{fig:Sec1_lambdas} are given in Figure \ref{fig:Sec2_gauges}, alongside the coordinates $(w/\lambda(w),(1-w)/\lambda(w))$ for all $w \in [0,1]$; these coordinates represent the relationship between $G$ and the ADF via equation \eqref{eqn:gauge_lambda}. One can again observe the variety in shapes. For the asymptotically dependent logistic copula, we have that $(1,1) \in \partial G$; this is true for all asymptotically dependent bivariate random vectors with a limit set. 

\begin{figure}[!h]
    \centering
    \includegraphics[width=\textwidth]{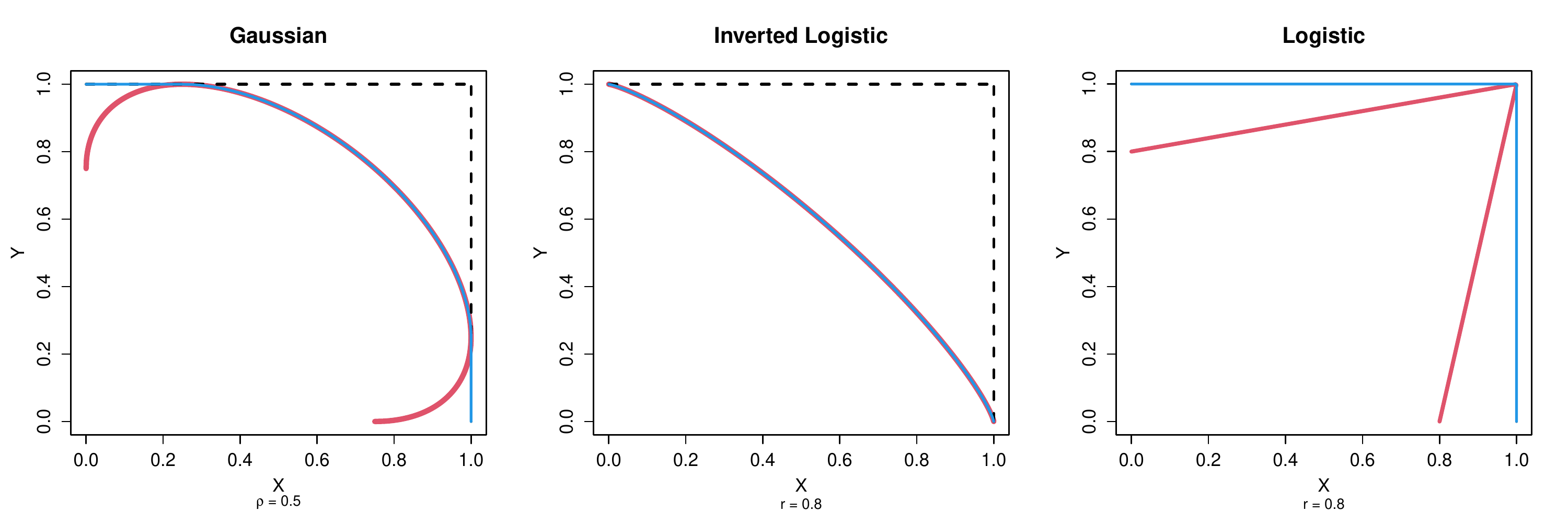}
    \caption{The boundary set $\partial G$ (given in red) for three example copulas, with coordinate limits denoted by the black dotted lines and the blue lines representing the coordinates $(w/\lambda(w),(1-w)/\lambda(w))$ for all $w \in [0,1]$. Left: bivariate Gaussian copula with coefficient $\rho = 0.5$. Centre: inverted logistic copula with dependence parameter $r = 0.8$. Right: logistic copula with dependence parameter $r = 0.8$.}
    \label{fig:Sec2_gauges}
\end{figure}

In practice, both the limit set, $G$, and its boundary, $\partial G$, are unknown. \citet{Simpson2022} propose an estimator for $\partial G$, which is then used to derive an estimator $\hat{\lambda}_{ST}$ for the ADF via equation \eqref{eqn:gauge_lambda}. The resulting estimator $\hat{\lambda}_{ST}$ was shown to outperform $\hat{\lambda}_{H}$ in a wide range of scenarios \citep{Simpson2022}. 

Estimation of $\partial G$ uses an alternative radial-angular decomposition of $(X,Y)$, with $R^*:= X+Y$ and $V^*:=X/(X+Y)$. \citet{Simpson2022} assume the tail of $R^* \mid V^*=v^*$, $v^* \in [0,1]$, follows a generalised Pareto distribution \citep{Davison1990} and then use generalised additive models to capture trends over angles in both the threshold and generalised Pareto distribution scale parameter \citep{Youngman2019}. Next, high quantile estimates from the conditional distributions $R^* \mid V^*=v^*$, $v^* \in [0,1]$ are computed using the fitted generalised Pareto distributions. They are then transformed back to the original scale using $X=R^*V^*$ and $Y = R^*(1-V^*)$ and finally scaled onto the set $[0,1]^2$ to give an estimate of $\partial G$; see \citet{Simpson2022} for further details.  

\citet{Wadsworth2023} also provide methodology for estimation of $\partial G$, though their focus is on estimation of tail probabilities more generally, including in dimensions greater than two. Furthermore, their approach requires prior selection of a parametric form for $g$. We therefore restrict our attention to the work of \citet{Simpson2022} as their main focus is semi-parametric estimation for $\partial G$ in two dimensions.

When applying the estimators $\hat{\lambda}_{H}$ and $\hat{\lambda}_{ST}$ in Section \ref{Sec4} and \ref{Sec5}, we use the tuning parameters suggested in the original approaches. In the case of $\hat{\lambda}_H$, we set $u$ to be the empirical 90\% quantile of $T_w$. The default tuning parameters for $\hat{\lambda}_{ST}$ can be found in \citet{Simpson2022}, and example estimates of the set $\partial G$ obtained using the suggested parameters are given in the Supplementary Material. For calculating this estimator, we used the code available at \url{https://github.com/essimpson/self-consistent-inference}.

\section{Novel theoretical results pertaining to the ADF} \label{sec:theory}
In this section, we outline some new theoretical results on the ADF. These results impose shape constraints that the ADF must satisfy, and follow from the properties of $\kappa$ introduced in Section \ref{Sec1}. 

\begin{proposition} \label{prop:kappa_set}
    For any $w_1,w_2 \in [0,1]$ such that $w_1 \leq w_2$, we have 
    \begin{equation*}
        w_1/\lambda(w_1) \leq w_2/\lambda(w_2) \hspace{1em} \text{and} \hspace{1em} (1-w_1)/\lambda(w_1) \geq (1-w_2)/\lambda(w_2). 
    \end{equation*}
\end{proposition}
\begin{proof}
    For the first statement, the proof is trivial if at least one of $w_1$ or $w_2$ equals $0$. Therefore, without loss of generality, assume that $w_1,w_2>0$. Recall $\lambda(w) = \kappa(w,1-w)$, and that $\kappa$ is homogeneous of order 1 and monotonic in both arguments. Setting $t := w_1/w_2 \in (0,1]$, we have
    \begin{align*}  
        \frac{w_1}{w_2}\lambda(w_2) &= t\kappa(w_2,1-w_2) \\
        &= \kappa(w_1,t(1-w_2)) \\
        &\leq \kappa(w_1,t(1-w_1)) \\
        &\leq \kappa(w_1,1-w_1) \\
        &= \lambda(w_1),
    \end{align*}
    implying $w_1/\lambda(w_1) \leq w_2/\lambda(w_2)$. For the second argument, the proof is again trivial if $w_2 = 1$ or $w_1 = w_2 = 1$, so assume $w_1,w_2 < 1$. Setting $t:= (1-w_2)/(1-w_1) \in (0,1]$, similar reasoning shows that $(1-w_1)/\lambda(w_1) \geq (1-w_2)/\lambda(w_2)$, completing the proof.

\end{proof}


\begin{remark}
    The intuition behind Proposition \ref{prop:kappa_set} comes from considering the relationship between the ADF and the boundary set $\partial G$ described in equation \eqref{eqn:gauge_lambda}. Considering the sets described by the blue lines in Figure \ref{fig:Sec2_gauges}, this implies that the $x$-coordinates ($y$-coordinates) of the sets must be increasing (decreasing) as the ray $w$ increases, as is clear from the figure.   
\end{remark}

Proposition \ref{prop:kappa_set} also implies several interesting properties that the ADF must satisfy.
\begin{corollary} \label{prop:lam_shape}
    Suppose there exists $w^* \in [0,0.5]$, or $w^* \in [0.5,1]$, such that $\lambda(w^*) = \max(w^*,1-w^*)$. Then $\lambda(w) = \max(w,1-w)$ for all $w \in [0,w^*]$, or $w \in [w^*,1]$.
\end{corollary}
\begin{proof}
    Considering the first case $w^* \in [0,0.5]$, Proposition \ref{prop:kappa_set} gives that $(1-w)/\lambda(w) \geq (1-w^*)/\lambda(w^*) = (1-w^*)/(1-w^*) = 1$ for any $w \in [0,w^*]$, implying $\max(w,1-w) = (1-w) \geq \lambda(w)$. Since $\lambda(w) \geq \max(w,1-w)$, we must therefore have $\lambda(w) = \max(w,1-w)$. The same reasoning applies for $w \in [w^*,1]$ when $w^* \geq 0.5$ and $\lambda(w^*) = \max(w^*,1-w^*)$.
\end{proof}
Corollary \ref{prop:lam_shape} states that if the ADF equals the lower bound for any angle in the interval $[0,0.5]$ (or the interval $[0.5,1]$), then it must also equal the lower bound for all angles less (greater) than this angle. This has further implications when we consider the conditional extremes modelling framework described in equation \eqref{eqn:heff_tawn}. Let $a_{y\mid x}$ and $a_{x\mid y}$ be the normalising functions for conditioning on the events $X>u$ and $Y>u$ respectively, and let $\alpha_{y\mid x} := \lim_{u \to \infty}a_{y\mid x}(u)/u$ and $\alpha_{x\mid y} := \lim_{u \to \infty}a_{x\mid y}(u)/u$, with $\alpha_{y\mid x},\alpha_{x\mid y} \in [0,1]$. From \citet{Nolde2022}, we have that $g(1,\alpha_{y \mid x}) = 1$ and $g(\alpha_{x \mid y},1) = 1$, with $g$ defined as in equation \eqref{eqn:gauge}, and $\alpha_{y\mid x}$, $\alpha_{x\mid y}$ are the maximum such values satisfying these equations. Assuming that the values of $\alpha_{y \mid x}$ and $\alpha_{x \mid y}$ are known, we have the following result. 
\begin{corollary} \label{prop:cond_ext_lam}
    For all $w \in [0,\alpha^*_{x \mid y}] \bigcup [\alpha^*_{y \mid x},1]$, with $\alpha^*_{x \mid y}:=\alpha_{x \mid y}/(1+\alpha_{x \mid y})$ and $\alpha^*_{y \mid x}:=1/(1+\alpha_{y \mid x})$, we have $\lambda(w) = \max(w,1-w)$.  
\end{corollary}
\begin{proof}
    To begin, consider the ray $\alpha^*_{x \mid y} \in [0,0.5]$ and observe that $(\alpha_{x \mid y},1) \in \partial G$. From this, one can see that $(\alpha_{x \mid y},1) \in R_{\alpha^*_{x \mid y}}$. Equation \eqref{eqn:gauge_lambda} therefore implies that $r_{\alpha^*_{x \mid y}} = 1$, and hence $\lambda(\alpha^*_{x \mid y}) = \max(\alpha^*_{x \mid y},1-\alpha^*_{x \mid y})$. From Corollary \ref{prop:lam_shape}, it follows that $\lambda(w) = \max(w,1-w)$ for all $w \in [0,\alpha^*_{x \mid y}]$. Considering the ray $\alpha^*_{y \mid x} \in [0.5,1]$ in an analogous manner, the result follows.  
\end{proof}
Corollaries \ref{prop:lam_shape} and \ref{prop:cond_ext_lam} are illustrated in Figure \ref{fig:gauge_lambda_result} for a Gaussian copula with $\rho=0.5$. Here, $\alpha_{x \mid y}=\alpha_{y \mid x} = 0.25$, implying $\lambda(w) = \max(w,1-w)$ for all $w \in [0,0.2] \bigcup [0.8,1]$; these rays correspond to the blue lines in the figure. One can observe that for any region $R_w$ defined along either of the blue lines (such as the shaded regions illustrated for $w = 0.1$ and $w = 0.9$), we have that $r_w = 1$, since these regions will intersect $\partial G$ at either the coordinates $(0.25,1)$ or $(1,0.25)$.   

\begin{figure}[!h]
    \centering
    \includegraphics[width=.6\textwidth]{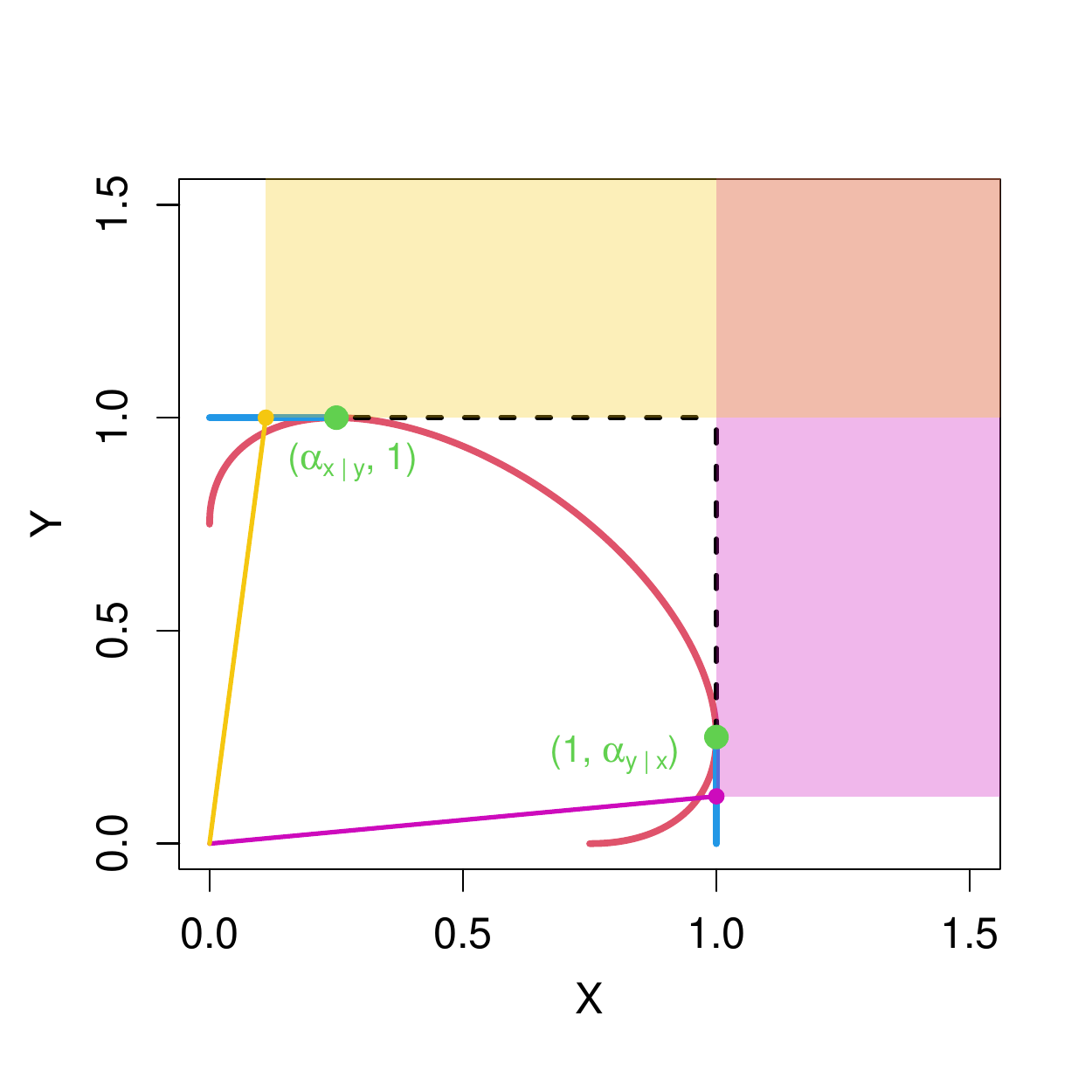}
    \caption{Pictorial illustration of the results described in Corollaries \ref{prop:lam_shape} and \ref{prop:cond_ext_lam}. The boundary set $\partial G$, given in red, is from the bivariate Gaussian copula with $\rho=0.5$, with the points $(1,\alpha_{y \mid x})$ and $(\alpha_{x \mid y},1)$ given in green. The blue lines represent the rays $w \in [0,\alpha^*_{x \mid y}] \bigcup [\alpha^*_{y \mid x},1]$, while the yellow and pink shaded regions represent the set $R_w$ for $w = 0.1$ and $w=0.9$, respectively. }
    \label{fig:gauge_lambda_result}
\end{figure}

Finally, Proposition \ref{prop:kappa_set} also implies constraints on the derivative of the ADF, where this exists. 

\begin{corollary} \label{corol:deriv}
    Let $\lambda'(w)$ denote the derivative of the ADF at $w$, where it exists. For all $w \in (0,1)$, we have that 
    \begin{equation*}
        -\lambda(w)/(1-w) \leq \lambda'(w) \leq \lambda(w)/w.   
    \end{equation*}
\end{corollary}
\begin{proof}
    Given any $w \in (0,1)$ and $\epsilon > 0$ such that $\epsilon \leq \min\{w, 1-w\}$, Proposition \ref{prop:kappa_set} implies that $w/\lambda(w) \leq (w+\epsilon)/\lambda(w+\epsilon)$ and $(w-\epsilon)/\lambda(w-\epsilon) \leq w/\lambda(w)$, which can be rearranged to give
    \begin{equation*}
        \frac{\lambda(w+\epsilon) - \lambda(w)}{\epsilon} \leq \lambda(w)/w, \quad \frac{\lambda(w) - \lambda(w-\epsilon)}{\epsilon} \leq \lambda(w)/w.
    \end{equation*} 
    Taking the limits as $\epsilon \to 0^+$, we obtain $\lambda'(w) \leq \lambda(w)/w$. By rearranging the equations $(1-w)/\lambda(w) \geq (1 - (w+\epsilon))/\lambda(w+\epsilon)$ and $(1 - (w-\epsilon))/\lambda(w-\epsilon)\geq (1-w)/\lambda(w)$ and taking limits in a similar manner, the result follows.  
    
\end{proof}

\begin{remark}
    Note that Corollary \ref{corol:deriv} corresponds to the same derivative constraints originally derived in \citet{Wadsworth2013}, albeit with a different proof. 
\end{remark}


The results introduced in this section provide necessary conditions that the ADF must satisfy. Furthermore, each of the introduced corollaries follow directly from Proposition \ref{prop:lam_shape}; therefore, if an estimator satisfies the conditions of this proposition, it will automatically satisfy the remaining conditions. We therefore incorporate Proposition \ref{prop:lam_shape} into our estimation framework for $\lambda$; see Section \ref{Sec3} for further details. We remark that only one of the existing approaches for estimating the ADF \citep{Murphy-Barltrop2022} has considered such constraints, although they will automatically be satisfied by ADF estimators derived from valid limit set estimates.



\section{Novel estimators for the ADF} \label{Sec3}
Motivated by the goal of global estimation, we propose a range of novel estimators for the ADF. We recall that the ADF and Pickands' dependence function exhibit several theoretical similarities, as listed in Section \ref{Sec1}, and arise as exponential rate parameters for suitable constructions of structure variables \citep{Mhalla2019a}. We therefore begin by reviewing estimation of the Pickands' dependence function.

Because smooth functional estimation for the ADF is desirable, we restrict our review to approaches for the Pickands' dependence function which achieve this: notably spline-based techniques \citep{Hall2000,Cormier2014} and techniques that utilise the family of Bernstein–B\'ezier polynomials \citep{Guillotte2016,Marcon2016,Marcon2017a}. In this paper, we focus on to the latter category, since spline-based techniques typically result in more complex formulations and a larger number of tuning parameters. Moreover, approaches based on Bernstein–B\'ezier polynomials have been shown to improve estimator performance across a wide range of copula examples \citep{Vettori2018}. For estimation of Pickands' dependence function, the following family of functions is considered
\begin{equation*}
    \mathcal{B}_k = \left\{ \sum_{i=0}^k \beta_i \binom{k}{i} w^i(1-w)^{k-i} \; \Bigg\vert \; \pmb{\beta} \in [0,1]^{k+1}, \; w \in [0,1] \right\},
\end{equation*}
where $k \in \NN$ denotes the polynomial degree. Note that $\mathcal{B}_k$ is a sub-family from the class of Bernstein–B\'ezier polynomials. Many approaches assume that the Pickands' dependence function $A \in \mathcal{B}_k$ and propose techniques for estimating the coefficient vector $\pmb{\beta}$, resulting in an estimator $\hat{\pmb{\beta}}$. This automatically ensures $A(t) \leq 1$ for all $t \in [0,1]$, thereby satisfying the theoretical upper bound of the Pickands' dependence function. 


We make a similar assumption about the ADF, and use this to propose novel estimators. However, unlike the Pickands' dependence function, the ADF is unbounded from above, meaning functions in $\mathcal{B}_k$ cannot represent all forms of extremal dependence captured by equation \eqref{eqn:wads_tawn}. Moreover, the endpoint conditions $\lambda(0) = \lambda(1) = 1$ are not necessarily satisfied by functions in $\mathcal{B}_k$. We therefore propose an alternative family of polynomials: given $k \in \NN$, let 
\begin{align} %
    \begin{split} \label{eqn:new_BP}
        \mathcal{B}^*_k = \Bigg\{ (1-w)^k + \sum_{i=1}^{k-1} \beta_i \binom{k}{i} &w^i(1-w)^{k-i} + w^k =: f(w) \; \Bigg\vert \; w \in [0,1], \pmb{\beta} \in [0,\infty)^{k-1} \Bigg\}.   
    \end{split}
\end{align}
Functions in this family are unbounded from above, and $f(0) = f(1) = 1$ for all $f \in \mathcal{B}^*_k$. Note that functions in $\mathcal{B}^*_k$ are not required to satisfy the lower bound of the ADF; this bound is instead imposed in a post-processing procedure, as detailed in Section \ref{Subsec:theory_impose}.


For the remainder of this section, let $\lambda(\cdot ; \; \pmb{\beta}) \in \mathcal{B}^*_k$ represent a form of the ADF from $\mathcal{B}^*_k$. Interest now lies in estimating the coefficient vector $\pmb{\beta}$, which requires choice of the degree $k \in \NN$. This is a trade-off between flexibility and computational complexity; polynomials with small values of $k$ may not be flexible enough to capture all extremal dependence structures, resulting in bias, while high values of $k$ will increase computational burden and parameter variance. 

\subsection{Composite likelihood approach} \label{Subsec3.1}
One consequence of equation \eqref{eqn:cond_wads_tawn} is that, for all $w \in [0,1]$, the conditional variable $T^*_w := (T_w - u_w \mid T_w > u_w) \sim$ Exp($\lambda(w)$), approximately, for large $u_w$. The density of this variable is $f_{T^*_w}(t^*_w) \approx \lambda(w)e^{-\lambda(w)t^*_w}, t^*_w>0$, resulting in a likelihood function for min-projection exceedances of $u_w$. Let $(\mathbf{x},\mathbf{y}) := \{ (x_i,y_i): \;i = 1, \hdots, n\}$ denote $n$ independent observations from the joint distribution of $(X,Y)$. For each $w \in \mathcal{W}$, where $\mathcal{W}$ denotes some finite subset spanning the interval $[0,1]$, let $\mathbf{t}_w := \{ \min(x_i/w,y_i/(1-w)): \; i = 1, \hdots, n\}$ and take $u_w$ to be the empirical $q$ quantile of $\mathbf{t}_w$, with $q$ close to $1$ and fixed across $w$. Letting $\mathbf{t}^*_w := \{ t_w - u_w \mid  t_w \in \mathbf{t}_w, t_w > u_w \}$, we have a set of realisations from the conditional variable $T^*_w$. 

One approach to obtain an estimate of $\lambda(w)$ while considering all $w \in \mathcal{W}$ simultaneously is to use a composite likelihood, in which multiple likelihood components are treated as independent whether or not they are independent. Provided each component is a valid density function, the resulting likelihood function provides unbiased parameter estimates under the true model; see \citet{Varin2011} for further details. For this model, the likelihood function is
\begin{equation} \label{eqn:cl_function}
    \mathcal{L}_C(\pmb{\beta}) = \prod_{w \in \mathcal{W}} \prod_{t^*_w \in \mathbf{t}^*_w} \lambda(w ; \; \pmb{\beta}) e^{-\lambda(w ; \; \pmb{\beta}) t^*_w} = \left[\prod_{w \in \mathcal{W}}\lambda(w;\; \pmb{\beta})^{\lvert \mathbf{t}^*_w \rvert}\right] \times e^{-\sum_{w \in \mathcal{W}}\sum_{t^*_w \in \mathbf{t}^*_w}\lambda(w;\; \pmb{\beta})t_w},
\end{equation}
where $\lvert \mathbf{t}^*_w \rvert$ denotes the cardinality of the set $\mathbf{t}^*_w$. This composite likelihood function has equal weights across all $w \in \mathcal{W}$ (the `components'). An estimator of the ADF, $\hat{\lambda}_{CL}$, is $\lambda(\cdot; \; \hat{\pmb{\beta}}_{CL})$, where $\hat{\pmb{\beta}}_{CL} := \argmax_{\pmb{\beta} \in [0,\infty)^{k-1}} \mathcal{L}_C(\pmb{\beta})$. 

To apply this method in practice, one must first specify a set $\mathcal{W}$ and a probability $q$. Given some large, odd-valued $m \in \NN$, we let $\mathcal{W} := \{ i/(m-1) : i = 0,1,\cdots,m-1 \}$; this corresponds to a set of equally spaced rays $\mathcal{W}$ spanning the interval $[0,1]$, with $\{0,0.5,1 \} \subset \mathcal{W}$. Selection of $m$ and $q$ are discussed in Section \ref{Subsec3.4}. The former is akin to selecting the degree of smoothing, while the latter is analogous to selecting a threshold for the generalised Pareto distribution defined in Section \ref{Sec2} in the univariate setting. 

\subsection{Probability ratio approach}\label{Subsec3.2}
With $\mathcal{W}$ and $\mathbf{t}_w$ defined as in Section \ref{Subsec3.1}, consider two probabilities $q<p<1$, both close to one. Given any $w \in \mathcal{W}$, let $u_w$ and $v_w$ denote the $q$ and $p$ empirical quantiles of $\mathbf{t}_w$, respectively. Assuming the distribution function of $T_w$ is strictly monotonic, equation \eqref{eqn:cond_wads_tawn} implies that 
\begin{equation} \label{eqn:quant_ratio}
    \frac{1-p}{1-q} = \Pr( T_w > v_w \mid T_w > u_w) \approx e^{-\lambda(w)(v_w - u_w)} \Rightarrow \left\vert \frac{1-p}{1-q} - e^{-\lambda(w)(v_w - u_w)} \right\vert \approx 0,   
\end{equation}
Similarly to \citet{Murphy-Barltrop2022}, we exploit equation \eqref{eqn:quant_ratio} to obtain an estimator for the ADF. Firstly, we observe that this representation holds for all $w \in \mathcal{W}$, hence 
\begin{equation*}
    \sum_{w \in \mathcal{W}}\left\vert \frac{1-p}{1-q} - e^{-\lambda(w)(v_w - u_w)} \right\vert \approx 0. 
\end{equation*}
To ensure equation  \eqref{eqn:quant_ratio} holds requires careful selection of $q$ and $p$. This selection also represents a bias-variance trade off: probabilities too small (big) will induce bias (high variability). Moreover, owing to the different rates of convergence to the limiting ADF, a single pair $(q,p)$ is unlikely to be appropriate across all data structures. We instead consider a range of probability pairs simultaneously. Specifically, letting $\{(q_{j},p_{j}) \mid q_{j} < p_{j} < 1, 1 \leq j \leq h\}$, $h \in \NN$, be pairs of probabilities near one, consider the expression
\begin{equation*}
    S(\pmb{\beta}) := \sum_{w \in \mathcal{W}}\sum_{j=1}^h \left\vert   \left[\frac{1-p_j}{1-q_j} \right] - e^{-\lambda(w;\; \pmb{\beta})(v_{w,j} - u_{w,j})} \right\vert, 
\end{equation*}
in which $u_{w,j}$ and $v_{w,j}$ denote $q_j$ and $p_j$ empirical quantiles of $\mathbf{t}_w$, respectively, for each $j = 1, \hdots, h$. Since minimising $S(\pmb{\beta})$ should provide a reasonable estimate  of $\lambda$ over all $\mathcal{W}$, we set $\hat{\pmb{\beta}}_{PR} = \argmin_{\pmb{\beta} \in [0,\infty)^{k-1}} S(\pmb{\beta})$ and denote by $\hat{\lambda}_{PR}$ the estimator $\lambda(\cdot;\; \hat{\pmb{\beta}}_{PR})$. Similarly to $\hat{\lambda}_{CL}$, one must select the sets $\mathcal{W}$ and $\{(q_{j},p_{j}) \mid q_{j} < p_{j} < 1, 1 \leq j \leq h\}$, $h \in \NN$ prior to applying this estimator; see Section \ref{Subsec3.4}. 

\subsection{Incorporating knowledge of conditional extremes parameters}\label{Subsec3.3}
Assuming we know the conditional extremes parameters $\alpha_{y\mid x},\alpha_{x\mid y}$, Corollary \ref{prop:cond_ext_lam} implies that for all $w \in [0,\alpha^*_{x \mid y}] \bigcup [\alpha^*_{y \mid x},1]$, with $\alpha^*_{x \mid y}=\alpha_{x \mid y}/(1+\alpha_{x \mid y})$ and $\alpha^*_{y \mid x}=1/(1+\alpha_{y \mid x})$, $\lambda(w) = \max(w,1-w)$.  In this section, we exploit this result to improve estimation of the ADF.

In practice, $\alpha^*_{x \mid y}$ and $\alpha^*_{y \mid x}$ are unknown; however, estimates $\hat{\alpha}_{y \mid x}$ and $\hat{\alpha}_{x \mid y}$ are commonly obtained using a likelihood function based on a misspecified model for the distribution $D$ in equation \eqref{eqn:heff_tawn} \citep[e.g., ][]{Jonathan2014}. The resulting estimates, denoted $\hat{\alpha}^*_{x \mid y}$, $\hat{\alpha}^*_{y \mid x}$, can be used to approximate the ADF for $w \in [0,\hat{\alpha}^*_{x \mid y}) \bigcup (\hat{\alpha}^*_{y \mid x},1]$. What now remains is to combine this with an estimator for $\lambda(w)$ on $ [\hat{\alpha}^*_{x \mid y},\hat{\alpha}^*_{y \mid x}]$. 

A crude way to obtain an estimator via this framework would be to set $\lambda(w) = \max(w,1-w)$ for $w \in [0,\hat{\alpha}^*_{x \mid y}) \bigcup (\hat{\alpha}^*_{y \mid x},1]$ and $\lambda(w) = \hat{\lambda}_{H}(w)$, $\hat{\lambda}_{CL}(w)$ or $\hat{\lambda}_{PR}(w)$ for $w \in [\hat{\alpha}^*_{x \mid y},\hat{\alpha}^*_{y \mid x}]$. However, this results in discontinuities at $\hat{\alpha}^*_{x \mid y}$ and $\hat{\alpha}^*_{y \mid x}$. Instead, for the smooth estimators, we rescale $\mathcal{B}_k^*$ such that the resulting ADF estimate is continuous for all $w \in [0,1]$. Consider the set of polynomials
\begin{align} %
    \begin{split} \label{eqn:heff_tawn_BP}
         \Tilde{\mathcal{B}}_k = \Bigg\{  (1-\hat{\alpha}^*_{x \mid y}) &\left(1- \frac{v-\hat{\alpha}^*_{x \mid y}}{\hat{\alpha}^*_{y \mid x}-\hat{\alpha}^*_{x \mid y}} \right)^k + \sum_{i=1}^{k-1} \beta_i \binom{k}{i} \left(\frac{v-\hat{\alpha}^*_{x \mid y}}{\hat{\alpha}^*_{y \mid x}-\hat{\alpha}^*_{x \mid y}}\right)^i\left(1-\frac{v-\hat{\alpha}^*_{x \mid y}}{\hat{\alpha}^*_{y \mid x}-\hat{\alpha}^*_{x \mid y}}\right)^{k-i} + \\
        & \hat{\alpha}^*_{y \mid x}\left(\frac{v-\hat{\alpha}^*_{x \mid y}}{\hat{\alpha}^*_{y \mid x}-\hat{\alpha}^*_{x \mid y}}\right)^k =:f(v) \; \Bigg\vert \; v \in [\hat{\alpha}^*_{x \mid y},\hat{\alpha}^*_{y \mid x}], \; \pmb{\beta} \in [0,\infty)^{k-1} \Bigg\}.   
    \end{split}
\end{align}
For all $f \in \Tilde{\mathcal{B}}_k$, we have that $f(\hat{\alpha}^*_{x \mid y}) = (1-\hat{\alpha}^*_{x \mid y})$ and $f(\hat{\alpha}^*_{y \mid x}) = \hat{\alpha}^*_{y \mid x}$, and each $f$ is only defined on the interval $[\hat{\alpha}^*_{x \mid y},\hat{\alpha}^*_{y \mid x}]$. Letting $\Tilde{\lambda}(\cdot \; ; \pmb{\beta}) \in \Tilde{\mathcal{B}}_k$ represent a form of the ADF for $w \in [\hat{\alpha}^*_{x \mid y},\hat{\alpha}^*_{y \mid x}]$, the techniques introduced in Sections \ref{Subsec3.1} and \ref{Subsec3.2} can be used to obtain estimates of the coefficient vectors, which we denote $\hat{\pmb{\beta}}_{CL2}$ and $\hat{\pmb{\beta}}_{PR2}$, respectively. The resulting estimators for $\lambda$ are
\begin{equation*}
    \hat{\lambda}_{CL2}(w) = \begin{cases}
        \Tilde{\lambda}(w;\hat{\pmb{\beta}}_{CL2}) \; & \text{for} \; w \in [\hat{\alpha}^*_{x \mid y},\hat{\alpha}^*_{y \mid x}], \\
        \max(w,1-w) \; & \text{for} \; w \in [0,\hat{\alpha}^*_{x \mid y}) \bigcup (\hat{\alpha}^*_{y \mid x},1],
    \end{cases} 
\end{equation*}
with $\hat{\lambda}_{PR2}$ defined analogously. We lastly define the discontinuous estimator $\hat{\lambda}_{H2}$ as
\begin{equation*}
    \hat{\lambda}_{H2} := \begin{cases}
        \hat{\lambda}_{H}(w) \; & \text{for} \; w \in [\hat{\alpha}^*_{x \mid y},\hat{\alpha}^*_{y \mid x}], \\
        \max(w,1-w) \; & \text{for} \; w \in [0,\hat{\alpha}^*_{x \mid y}) \bigcup (\hat{\alpha}^*_{y \mid x},1].
    \end{cases} 
\end{equation*}
This is obtained by combining the pointwise Hill estimator with the information provided by the estimates $\hat{\alpha}^*_{x \mid y},\hat{\alpha}^*_{y \mid x}$. Illustrations of all the estimators discussed in this section, as well as in Section \ref{Sec2}, can be found in the Supplementary Material. 

We note that estimation of the parameters, $\alpha_{x \mid y},\alpha_{y \mid x},$ is subject to uncertainty and will not give a perfect representation of $\alpha^*_{x \mid y}$ and $\alpha^*_{y \mid x}$. However, we believe it is worth exploring the quality of the resulting estimates, and find in Section \ref{Sec4} that they generally provide improvement over estimators that do not exploit this link. We also note the potential for independent estimates of $\lambda$ to be used for improving estimation of $\alpha_{x \mid y}$ and $\alpha_{y \mid x}$ in the conditional extremes model, but leave exploration of this to future work.

\subsection{Incorporating theoretical properties} \label{Subsec:theory_impose}
All estimators introduced so far are not required to satisfy the property of $\lambda$ introduced in Proposition \ref{prop:lam_shape}, or the lower bound on the ADF discussed in Section \ref{Sec1}. Furthermore, the estimator $\Hat{\lambda}_{H}$ is also not guaranteed to satisfy the endpoint conditions, i.e., $\lambda(0) = \lambda(1) = 1$.

There are several techniques one could use to impose these properties. For instance, one could incorporate penalty terms to objective functions $\mathcal{L}_C(\pmb{\beta})$ or $S(\pmb{\beta})$ to penalise for cases when the conditions of Proposition \ref{prop:lam_shape}, or the ADF lower bound, are not satisfied. Alternatively, one could also impose the theoretical properties of the ADF in post-processing steps; such a procedure can be applied regardless of whether the original estimator is smooth or local. In unreported results, we considered penalising the objective function $\mathcal{L}_C(\pmb{\beta})$ for violations of ADF properties, with larger penalties for greater violations. This approach allows for simpler optimisation of the objective function than imposing a very strong penalty for any violation, but introduces the choice of a penalty parameter and does not guarantee the properties are fully satisfied. We therefore opt instead to apply a post-processing procedure to each of the ADF estimators; this has the added advantage of also being applicable to the pointwise Hill estimators.


For any estimator $\hat{\lambda}_{-}$, assume that the set $\{ w \in [0,1] \mid \hat{\lambda}_{-}(w) < \max(w,1-w)\}$ is non-empty. To ensure the ADF is bounded from below, and satisfies the endpoint conditions $\lambda(0) = \lambda(1) = 1$, we set 
\begin{equation*}
    \hat{\lambda}_{-}(w) = \begin{cases}
        \max \left\{ \hat{\lambda}_{-}(w), \max(w,1-w) \right\} \; & \text{for} \; w \in (0,1), \\
        1 \; & \text{for} \; w \in \{0,1\}.
    \end{cases} 
\end{equation*}
Next, we ensure the conditions outlined in Proposition 
\ref{prop:lam_shape} are satisfied. Define the angular sets $\mathcal{W}^{\leq 0.5} = (w^{\leq 0.5}_1,w^{\leq 0.5}_2,\hdots,w^{\leq 0.5}_{(m-1)/2}) := \{ i/(m-1) : i = (m-1)/2,(m-1)/2 - 1,\hdots,0 \} \subset \mathcal{W}$ and $\mathcal{W}^{\geq 0.5} = (w^{\geq 0.5}_1,w^{\geq 0.5}_2,\hdots,w^{\geq 0.5}_{(m-1)/2}) := \{ i/(m-1) : i = (m-1)/2,(m-1)/2 + 1,\hdots,m-1 \} \subset \mathcal{W}$. We propose the following algorithm.

\begin{algorithm}
    \SetAlgoLined
    \caption{Algorithm for imposing Property \ref{prop:kappa_set}.} \label{algo:kappa_alg}
    
    \For{$i \leftarrow 2$ \KwTo $(m-1)/2$}{
        \If{$w^{\leq 0.5}_{i-1}/\hat{\lambda}_{-}(w^{\leq 0.5}_{i-1}) < w^{\leq 0.5}_{i}/\hat{\lambda}_{-}(w^{\leq 0.5}_{i})$}{
            set $\hat{\lambda}_{-}(w^{\leq 0.5}_{i}) := w^{\leq 0.5}_{i}\hat{\lambda}_{-}(w^{\leq 0.5}_{i-1})/w^{\leq 0.5}_{i-1}$ \;
        }
        \If{$(1-w^{\geq 0.5}_{i-1})/\hat{\lambda}_{-}(w^{\geq 0.5}_{i-1}) < (1-w^{\geq 0.5}_{i})/\hat{\lambda}_{-}(w^{\geq 0.5}_{i})$}{
            set $\hat{\lambda}_{-}(w^{\geq 0.5}_{i}) := (1-w^{\geq 0.5}_{i})\hat{\lambda}_{-}(w^{\geq 0.5}_{i-1})/(1-w^{\geq 0.5}_{i-1})$ \;
        }
         
    }
\end{algorithm}

This ensures the processed estimator $\hat{\lambda}_{-}$ satisfies the conditions of Properties \ref{prop:lam_shape} for $w \in \mathcal{W}$; this is the finite window that we use to represent $[0,1]$ in practice. This processing is applied to all pointwise and novel estimators, i.e., $\Hat{\lambda}_{H}$, $\Hat{\lambda}_{CL}, \Hat{\lambda}_{PR}$, $\Hat{\lambda}_{H2}$, $\Hat{\lambda}_{CL2}$ and $\Hat{\lambda}_{PR2}$, ensuring that the ADF estimates from each approach are always theoretically valid. In practice, we found that imposing these theoretical results also improved estimation quality within the resulting ADF estimates, both in terms of bias and variance. An illustration of the processing procedure is given in the Supplementary Material. 

\subsection{Tuning parameter selection} \label{Subsec3.4}

Prior to using any of the ADF estimators introduced in this section, we are required to select at least one tuning parameter. For the probability values required by the estimators introduced in Sections \ref{Subsec3.1} and \ref{Subsec3.2}, we set $q=0.90$, $\{q_{j} \}_{j=1}^h := \{0.87 + (j-1)\times 0.002\}_{j=1}^h$ and $p_{j} = q_{j} + 0.05$ for $j = 1, \ldots, h$, with $h=31$. These values were chosen to evaluate whether the resulting estimators improve upon the base estimator $\hat{\lambda}_H$ using (approximately) the same amount of tail information in all cases. We tested a range of probabilities for both estimators and found that the ADF estimates were not massively sensitive to these across different extremal dependence structures. For example, for $\hat{\lambda}_{CL}$, a lower $q$ resulted in mild improvements for asymptotically independent copulas, while simultaneously worsening the quality of ADF estimates for asymptotically dependent examples, while a higher $q$ led to higher variance. 

For the angular interval $\mathcal{W}$, we set $m = 1,000$, i.e., $\mathcal{W} = \{0, 0.001, 0.002, \hdots,0.999, 1\}$. This set was sufficient to ensure a high degree of smoothness in the resulting ADF estimates without too high a computational burden.  

For each of the novel estimators (except $\hat{\lambda}_{H2}$), we must also specify the degree $k \in \NN$ for the polynomial families described by equations \eqref{eqn:new_BP} and \eqref{eqn:heff_tawn_BP}. In the case of the Pickands' dependence function, studies have found that higher values of $k$ are preferable for very strong positive dependence, while the opposite is true for weak dependence \citep{Marcon2017a,Vettori2018}. We prefer to select a single value of $k$ that performs well across a range of dependence structures, while minimising the computational burden; this avoids the need to select this parameter when obtaining ADF estimates in practice. 

To achieve this objective, we estimated the root mean integrated squared error (RMISE), as defined in Section \ref{Subsec4.1}, of the estimators $\hat{\lambda}_{CL}$ and $\hat{\lambda}_{PR}$ with $k = 4,\hdots,11$ using $200$ samples from two Gaussian copula examples, corresponding to strong ($\rho=0.9$) and weak ($\rho=0.1$) positive dependence. Assessment of how the RMISE estimates vary over $k$ for both estimators suggests that $k=7$ is sufficient to accurately capture both dependence structures. The full results can be found in the Supplementary Material. We remark that this approach for selecting $k$ is somewhat ad hoc, and in practice, one could employ various diagnostic tools, such as the tool discussed in Section \ref{Subsec5.3}, to select $k$ on a case-by-case basis.


For each of the `combined' estimators in Section \ref{Subsec3.3}, we take the same tuning parameters as for the `non-combined' counterpart, since the combined estimators have near identical formulations only defined on a subset of $[0,1]$. For example, the empirical 90\% threshold of the min-projection is used for both $\hat{\lambda}_{H}$ and $\hat{\lambda}_{H2}$. Finally, when estimating the conditional extremes parameters, empirical $90\%$ conditioning thresholds are used.

\section{Simulation Study} \label{Sec4}

\subsection{Overview}\label{Subsec4.1}
In this section, we use simulation to compare the estimators proposed in Section \ref{Sec3} to the existing techniques described in Section \ref{Sec2}. For the comparison, we introduce nine copula examples, representing a wide variety of extremal dependence structures, and encompassing both extremal dependence regimes. 

The first three examples are from the bivariate Gaussian distribution, for which the dependence is characterised by the correlation coefficient $\rho \in [-1,1]$. We consider $\rho \in \{-0.6,0.1,0.6\}$, resulting in data structures exhibiting medium negative, weak positive, and medium positive dependence, respectively. Note that in the case of $\rho = -0.6$, the choice of exponential margins will hide the dependence structure \citep{Keef2013,Nolde2022}. 

For the next two examples, we consider the bivariate extreme value copula with logistic and asymmetric logistic families \citep{Gumbel1960,Tawn1988}. In both cases, the dependence is characterised by the parameter $r \in (0,1]$; we set $r = 0.8$, corresponding to weak positive extremal dependence. For the asymmetric logistic family, we also require two asymmetry parameters $(k_1,k_2) \in [0,1]^2$, which we fix to be $(k_1,k_2) = (0.3,0.7)$, resulting in a mixture structure. 

We next consider the inverted bivariate extreme value copula \citep{Ledford1997} for the logistic and asymmetric logistic families, with the dependence again characterised by the parameters $r$ and $(r,k_1,k_2)$, respectively. We set $r = 0.4$, corresponding to moderate positive dependence, and again fix $(k_1,k_2) = (0.3,0.7)$. Note that for this copula, the model described in equation \eqref{eqn:wads_tawn} is exact: see \citet{Wadsworth2013}.

Lastly, we consider the bivariate student t copula, for which dependence is characterised by the correlation coefficient $\rho \in [-1,1]$ and the degrees of freedom $\nu > 0$. We consider $\rho = 0.8$, $\nu = 2$ and $\rho = 0.2$, $\nu = 5$, corresponding to strong and weak positive dependence. 

Illustrations of the true ADFs for each copula are given in Figure \ref{fig:true_ADFs_copulas}, showing a range of extremal dependence structures. For examples where the ADF equals the lower bound, the copula exhibits asymptotic dependence. While the fifth copula exhibits asymmetric dependence, the limiting ADF is symmetric; the same is not true for its inverted counterpart.   

\begin{figure}[!h]
    \centering
    \includegraphics[width=\textwidth]{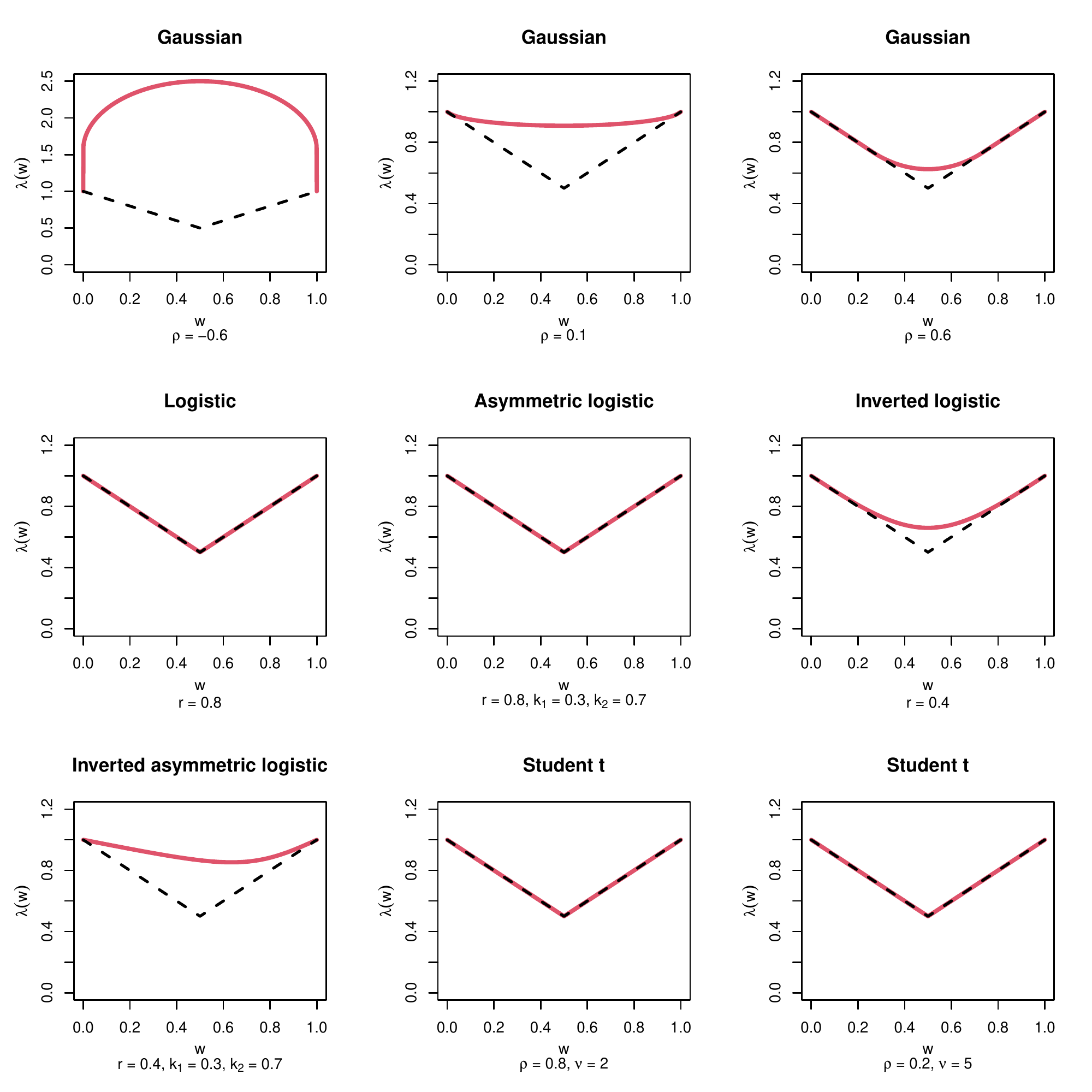}
    \caption{True ADFs (in red) for each copula introduced in Section \ref{Subsec4.1}, along with the theoretical lower bound (black dotted line).}
    \label{fig:true_ADFs_copulas}
\end{figure}

To evaluate estimator performance, we use the RMISE 
\begin{equation*}
    \text{RMISE}\left(\hat{\lambda}_{-}\right) = \left( \mathbb{E} \left[\int_{0}^1\left\{\hat{\lambda}_{-}(w) - \lambda(w)\right\}^2 \mathrm{d}w \right] \right)^{1/2},
\end{equation*}
where $\hat{\lambda}_{-}$ denotes an arbitrary estimator. Simple rearrangement shows that this metric is equal to the square root of the sum of integrated squared bias (ISB) and integrated variance (IV) \citep{Gentle2002}, i.e.,  
\begin{equation*}
    \text{RMISE}\left(\hat{\lambda}_{-}\right) = \left( \underbrace{\int_{0}^1\left[ \mathbb{E}\left\{\hat{\lambda}_{-}(w)\right\} - \lambda(w)\right]^2 \mathrm{d}w}_{\text{ISB}} + \underbrace{\int_{0}^1 \mathbb{E} \left( \left[ \mathbb{E} \left\{\hat{\lambda}_{-}(w)\right\} - \hat{\lambda}_{-}(w)\right]^2 \right) \mathrm{d}w}_{\text{IV}} \right)^{1/2}.
\end{equation*}
Therefore, the RMISE summarises the quality of an estimator in terms of both bias and variance, and can be used as a means to compare different estimators. 

\subsection{Results} \label{Subsec4.4}
For the copulas described in Section \ref{Subsec4.1}, data from each copula example was simulated on standard exponential margins with a sample size of $n = $10,000, and the integrated squared error (ISE) of each estimator was approximated for $1,000$ samples using the trapezoidal rule; see the Supplementary Material for further information. The square root of the mean of these estimates was then computed, resulting in a Monte–Carlo estimate of the RMISE.



The RMISE estimates for each estimator and copula combination are shown in Table \ref{table:RMISE_values}. Tables containing the corresponding Monte–Carlo error, ISB and IV values can be found in the Supplementary Material, along with root mean squared error estimates for individual rays $w \in \{0.1,0.3,0.5,0.7,0.9\}$. For each estimator, the bias varies significantly across the different copulas. However, in the majority of cases, the bias/variance are similar across most of the estimators. We remark that the magnitudes of the Monte–Carlo errors, as reported in the Supplementary Material, are small enough such that the ordering of RMISE estimates in Table \ref{table:RMISE_values} is likely to be a true reflection of the relative performance for each estimator.

\renewcommand{\arraystretch}{1.5}
\begin{table}[!h]

\caption{RMISE values (multiplied by 100) for each estimator and copula combination. Smallest RMISE values in each row are highlighted in bold, with values reported to 3 significant figures. Copulas 1-3 refer to the Gaussian copula with $\rho = -0.6$, $\rho = 0.1$ and $\rho = 0.6$, respectively. Copulas 4-7 refer to the logistic, asymmetric logistic, inverted logistic, inverted asymmetric logistic copulas, respectively. Copulas 8-9 refer to the student t copula with parameters $\rho = 0.8$, $\nu = 2$ and $\rho = 0.2$, $\nu = 5$, respectively.
\label{table:RMISE_values}}
\centering
\begin{tabular}[t]{cccccccc}
\toprule
Copula & $\hat{\lambda}_{H}$ & $\hat{\lambda}_{CL}$ & $\hat{\lambda}_{PR}$ & $\hat{\lambda}_{H2}$ & $\hat{\lambda}_{CL2}$ & $\hat{\lambda}_{PR2}$ & $\hat{\lambda}_{ST}$\\
\hline
Copula 1 & \textbf{61.2} & 61.3 & 66.2 & 61.4 & 61.9 & 66.7 & 63.7\\
\hline
Copula 2 & 3.44 & 3.36 & 3.67 & 3.41 & 3.35 & 3.66 & \textbf{2.95}\\
\hline
Copula 3 & 3.43 & 3.46 & 3.83 & 3.21 & 3.22 & 3.57 & \textbf{1.09}\\
\hline
Copula 4 & 4.58 & 4.71 & 6.89 & 4.24 & 4.24 & 6.17 & \textbf{2.77}\\
\hline
Copula 5 & 14 & 14.1 & 17.1 & 14 & 14.1 & 17.1 & \textbf{12.1}\\
\hline
Copula 6 & 2.05 & 2 & 2.18 & 1.78 & \textbf{1.75} & 1.92 & 2.12\\
\hline
Copula 7 & 2.79 & \textbf{2.68} & 2.93 & 2.77 & 2.69 & 2.93 & 3.96\\
\hline
Copula 8 & 1.04 & 1.05 & 1.44 & 0.562 & \textbf{0.535} & 0.72 & 1.87\\
\hline
Copula 9 & 11.9 & 12 & 14.9 & 11.9 & 12 & 14.9 & \textbf{11.1}\\
\bottomrule
\end{tabular}
\end{table}


While no estimator consistently outperforms the others, $\hat{\lambda}_{CL2}$ and $\hat{\lambda}_{ST}$ tend to have lower RMISE, ISB and IV values, on average. This is especially the case when comparing to the base estimator $\hat{\lambda}_H$. Furthermore, the `combined' estimators outperform their non-combined counterparts in many cases, suggesting that incorporating parameter estimates from the conditional extremes model can reduce bias and variance. The Gaussian copula with $\rho = -0.6$ has much higher RMISE values, indicating that none of the estimators capture negative dependence well, though this is in part due to the choice of exponential margins. 

We note that while the true ADFs are identical for the asymptotically dependent logistic, asymmetric logistic, and student t copulas, the corresponding RMISE values in Table \ref{table:RMISE_values} vary significantly. Notably, the asymmetric logistic, which possesses a complex asymmetric structure, has substantially higher RMISE values compared to the other asymptotically dependent copulas. We suspect these disparities arise at finite levels due to the different rates of convergence to the limiting ADF in equation \eqref{eqn:wads_tawn}, alongside the fact many multivariate extreme value models perform poorly in the case of asymmetric dependence \citep{Tendijck2021}.

Overall, these results indicate that no one estimator is preferable across all extremal dependence structures. However, we suggest using the estimators $\hat{\lambda}_{CL2}$ and $\hat{\lambda}_{ST}$ since, on average, these appeared to result in less bias and variance. The form of extremal dependence appears to affect the performance of both of these estimators; since this is often difficult to quantify a priori, we suggest using both estimators and evaluating relative performance via diagnostics, as we do in Section \ref{Sec5}. 

\section{Case Study} \label{Sec5}
\subsection{Overview} \label{Subsec5.1} 
Understanding the probability of observing extreme river flow events (i.e., floods) at multiple sites simultaneously is important in a variety of sectors, including insurance \citep{Keef2013a,Quinn2019,Rohrbeck2021a} and environmental management \citep{Lamb2010,Gouldby2017}. Valid risk assessments therefore require accurate evaluation of the extremal dependence between different sites. 

In this section, we estimate the ADF of river flow data sets obtained from gauges in the north of England, UK, which can be subsequently used to construct bivariate return curves. Daily average flow values ($m^3/s$) at six river gauge locations on different rivers were considered. The gauge sites are illustrated in Figure \ref{fig:riv_spatial_plot}. For each location, data is available between May 1993 and September 2020; however, we only consider dates where a measurement is available for every location. To avoid seasonality, we consider the interval October–March only; from our analysis, it appears that the highest daily flow values are observed in this period. This results in $n=4,659$ data points for each site. Plots of the daily flow time series can be found in the Supplementary Material; these plots suggest that an assumption of  stationarity is reasonable for the extremes of each data set. 

\begin{figure}[!h]
    \centering
    \includegraphics[width=.7\textwidth]{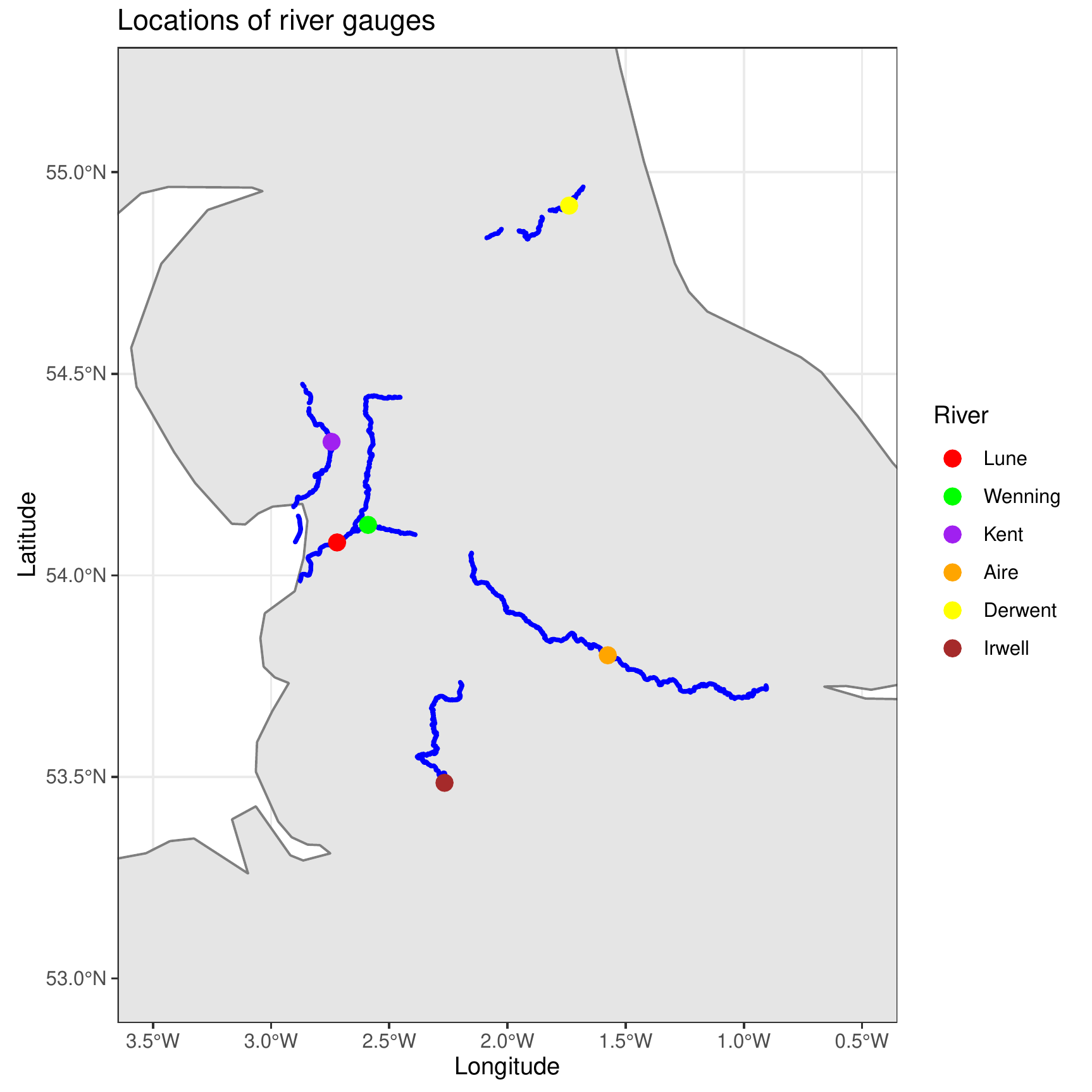}
    \caption{Locations of river gauges in the north of England, UK. Individual rivers illustrated in blue alongside the corresponding gauge locations.}
    \label{fig:riv_spatial_plot}
\end{figure} 

We fix the site on the river Lune to be our reference site and consider the extremal dependence between this and all other gauges. We first estimate the extremal dependence measure $\chi$ and the coefficient $\eta$ using the upper $10\%$ of the corresponding joint tails. Both $\chi$ and $\eta$ are limiting values; however, in practice, we are unable to evaluate such limits without a closed form for the joint distribution. We therefore calculate these values empirically. Taking $\chi$, for example, an estimate is $\hat{\chi}_{q} = \hat{\Pr}(X>\hat{x}_{q},Y>\hat{y}_{q})/\hat{\Pr}(X>\hat{x}_{q})$, where $\hat{\Pr}(\cdot)$ denotes an empirical probability estimate and $\hat{x}_{q}$ and $\hat{y}_{q}$ denote empirical $q$ quantile estimates for the variables $X$ and $Y$, respectively, and $q$ is some value close to $1$. Specifically, we take $q=0.9$. In practice, we are unlikely to observe $\chi = \hat{\chi}_{q}$, even at the most extreme quantile levels, i.e., as $q \to 1$. This can be problematic when trying to quantify the form of extremal dependence, since $\hat{\chi}_q>0$ may arise for asymptotically independent data sets (for example). Therefore, the estimated coefficients should act only as a rough guide for this quantification.

The dependence measure estimates and $95\%$ confidence intervals are shown in Figure \ref{fig:dep_coeff} as a function of distance from the reference site. Here and throughout, all confidence intervals are obtained via block bootstrapping with block size $b=40$; this value appears appropriate to account for the varying degrees of temporal dependence observed across the six gauge sites. These estimates suggest that asymptotic independence may exist for at least three of the site pairings; therefore, modelling techniques based on bivariate regular variation would likely fail to capture the observed extremal dependencies in this scenario.

\begin{figure}[!h]
    \centering
    \includegraphics[width=\textwidth]{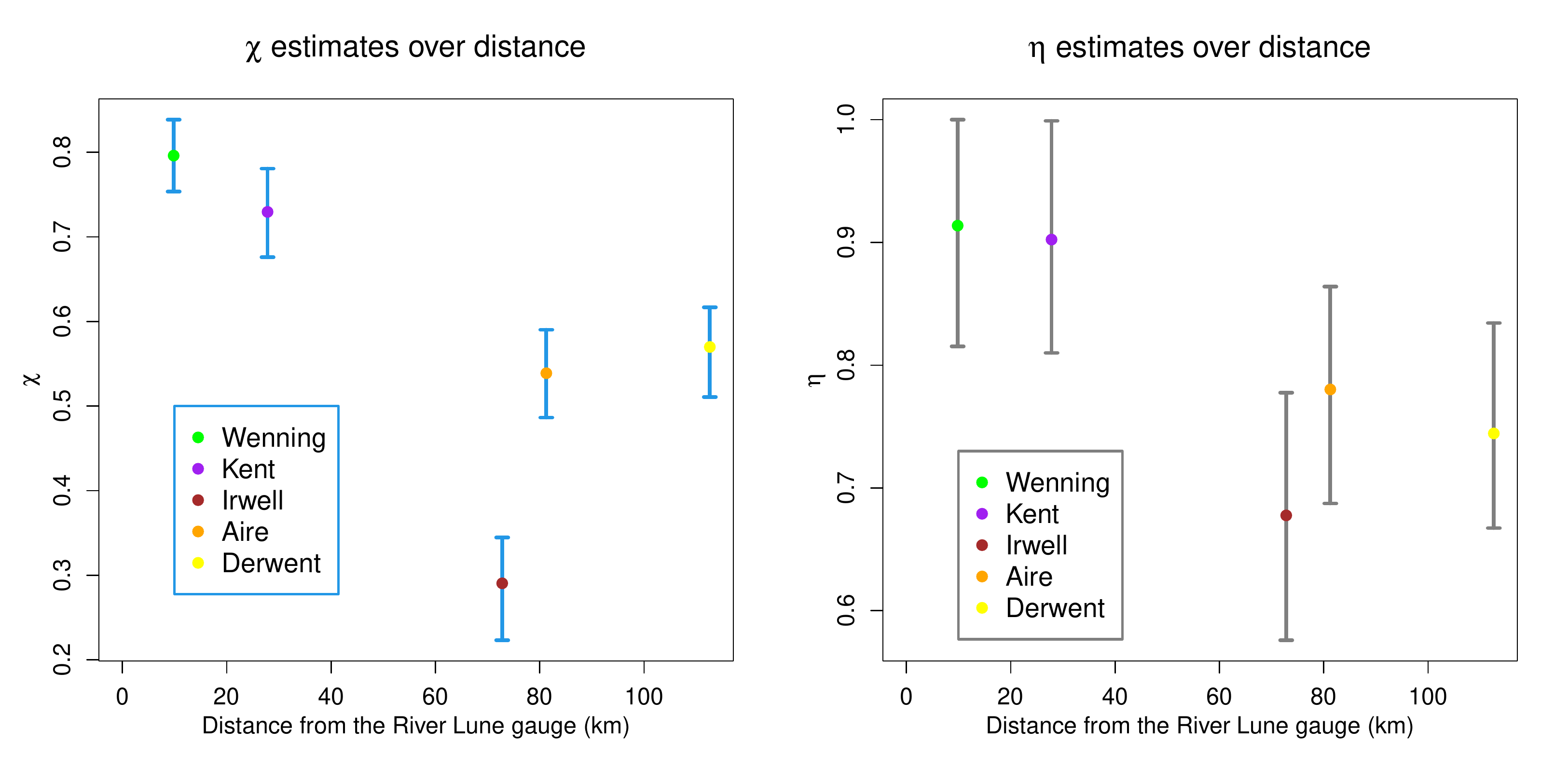}
    \caption{Estimated dependence coefficients as a function of distance from the Lune gauge, with $95\%$ pointwise confidence intervals given by the shaded regions. Left: Estimates of $\chi$ (blue). Right: Estimates of $\eta$ (grey).}
    \label{fig:dep_coeff}
\end{figure}

\subsection{ADF Estimation} \label{Subsec5.2} 

We transform each marginal data set to exponential margins using the semi-parametric approach of \citet{Coles1991}, whereby a generalised Pareto distribution is fitted to the upper tail and the body is modelled empirically. The generalised Pareto distribution thresholds are selected using the technique proposed in \citet{Murphy2023}. In spite of the data violating the independence assumption, diagnostic plots found in the Supplementary Material indicate reasonable model fits. Since our results from Section \ref{Sec4} suggest that the estimators $\hat{\lambda}_{CL2}$ and $\hat{\lambda}_{ST}$ perform best overall, we used these, alongside the base estimator $\hat{\lambda}_{H}$, to estimate the ADF for each combination of the reference gauge and the other five gauges. The resulting ADF estimates can be found in Figure \ref{fig:adf_ests}.

\begin{figure}[!h]
    \centering
    \includegraphics[width=\textwidth]{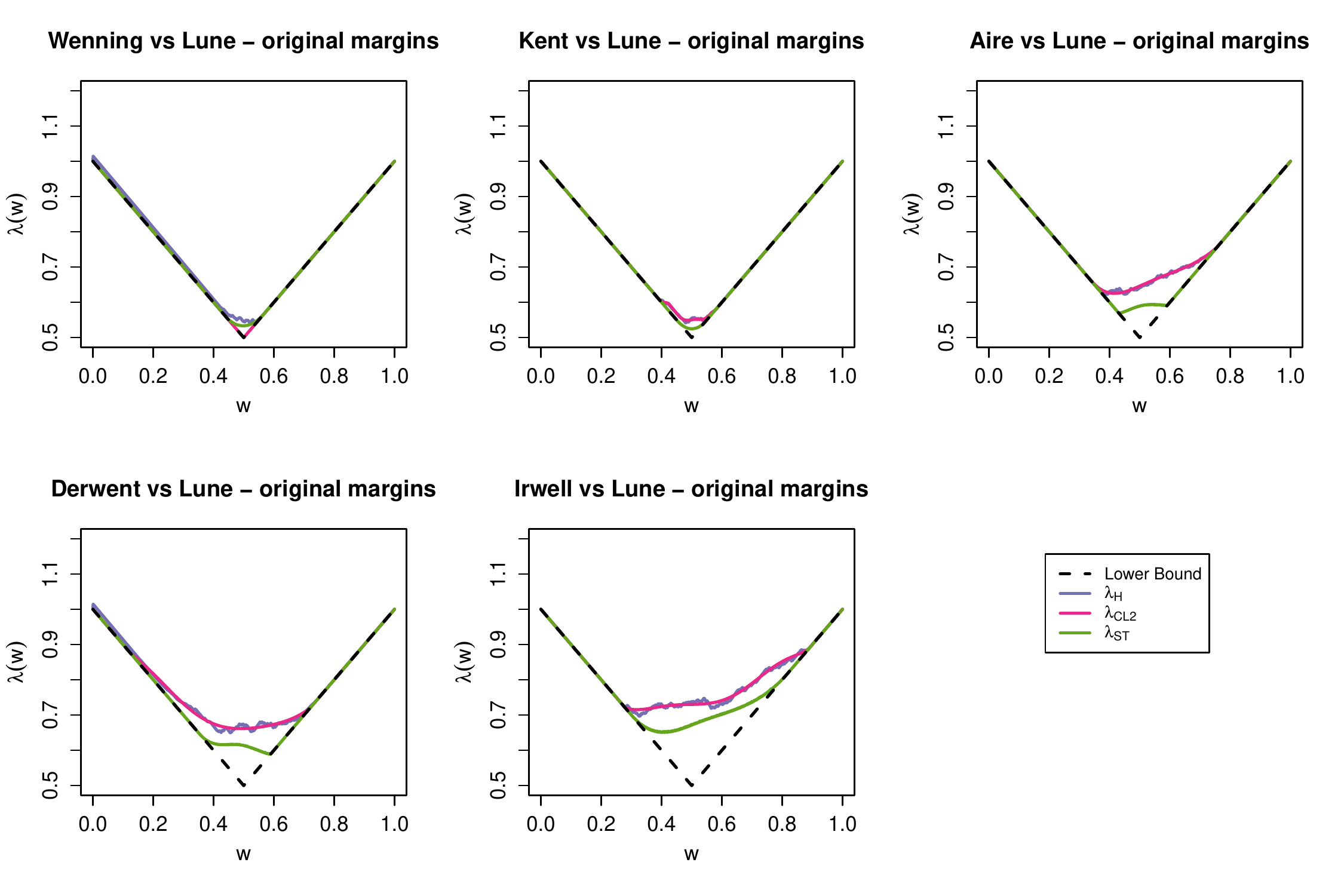}
    \caption{ADF estimates for each pair of gauge sites. The purple, pink and green lines represent the estimates from $\hat{\lambda}_{H}$, $\hat{\lambda}_{CL2}$ and $\hat{\lambda}_{ST}$, respectively, with the theoretical lower bound denoted by the black dotted lines.}
    \label{fig:adf_ests}
\end{figure}

One can observe contrasting shapes across the different pairs of gauges, illustrating the variety in the observed extremal dependence structures. These results illustrate that on the whole, the estimator $\hat{\lambda}_{CL2}$ is very much a smoothed version of $\hat{\lambda}_{H}$ on the interval $(\hat{\alpha}^*_{x \mid y}$, $\hat{\alpha}^*_{y \mid x})$, owing to the form of likelihood function used.

\subsection{Assessing goodness of fit for ADF estimates} \label{Subsec5.3} 

Recall that, from equation \eqref{eqn:cond_wads_tawn}, we have $T^*_w \sim \text{Exp}(\lambda(w))$ as $u_w \to \infty$ for all $w \in [0,1]$. We exploit this result to assess the goodness of fit for ADF estimates. 

Let $\hat{\lambda}(w), w \in [0,1]$, denote an estimated ADF obtained using the sample $\{ (x_i,y_i): \;i = 1, \hdots, n\}$. Given $w \in [0,1]$, let $u_w$ denote some high empirical quantile from the sample $\mathbf{t}_w$, and consider the sample $\mathbf{t}^*_w$, with $\mathbf{t}_w$ and $\mathbf{t}^*_w$ defined as in Section \ref{Subsec3.1}. If $\mathbf{t}^*_w$ is indeed a sample from an $\text{Exp}(\hat{\lambda}(w))$ distribution, we would expect good agreement between empirical and model quantiles. Analogously, $\hat{\lambda}(w)\mathbf{t}^*_w$ should represent a sample from an $\text{Exp}(1)$ distribution if $T^*_w \sim \text{Exp}(\lambda(w))$, independent of the ray $w$.

We use these results to derive local and global diagnostics for the ADF. First, for a subset of rays $w \in [0,1]$, corresponding to a range of joint survival regions, let $n_w = \lvert\mathbf{t}^*_w\rvert$ and $\mathbf{t}^*_{w(j)}$ denote the $j$-th order statistic of $\mathbf{t}^*_{w}$ and consider the set of pairs
\begin{equation*} 
    \left\{ \left( -\log(1-j/(n_w + 1))/\hat{\lambda}(w) , \mathbf{t}^*_{w(j)} \right) : j = 1, \hdots, n_w \right\}.
\end{equation*}
With $u_w$ fixed to be the 90\% empirical quantile of $\mathbf{t}_w$, quantile-quantile (QQ) plots for five individual rays, $w \in \{0.1,0.3,0.5,0.7,0.9\}$, are given in the first five panels Figure \ref{fig:adf_diag} for the third pair of gauges and the $\hat{\lambda}_{CL2}$ estimator. Uncertainty intervals are obtained via block bootstrapping on the set $\mathbf{t}^*_{w}$, i.e., the order statistics. We acknowledge a deficiency of this scheme that all sampled quantiles will be bounded by the interval $[\mathbf{t}^*_{w(1)},\mathbf{t}^*_{w(n_w)}]$. However, alternative uncertainty quantification approaches, such as parametric bootstraps or the use of beta quantiles for uniform order statistics, would fail to account for the observed temporal dependence within the data. 

For the global diagnostic, we propose the following: for each $i \in \{1,\hdots,n\}$, define the corresponding angular observation $w_i := x_i/(x_i + y_i)$ and let $u_{w_i}$ denote the 90\% empirical quantile of $\mathbf{t}_{w_i}$. Randomly sample one point $t^*$ from the set $\mathbf{t}^*_{w_i}$ and set $e_i := \lambda(w_i)(t^* - u_{w_i})$; repeating this process over all $i \in \{1,\hdots,n\}$, we obtain the set $\mathcal{E}:= \{ e_i \mid i \in \{1,\hdots,n\}\}$. We then consider the set of pairs $$\left\{ \left( -\log(1-j/(n + 1)) , e_{(j)} \right) : j = 1, \hdots, n \right\},$$ 
with $e_{(j)}$ denoting the $j$-th order statistic of $\mathcal{E}$; this comparison provides an overall summary for the quality of model fit across all angles. The corresponding QQ plot for the third pair of gauges and the $\hat{\lambda}_{CL2}$ estimator is given in the bottom right panel of Figure \ref{fig:adf_diag}, with uncertainty bounds obtained via block bootstrapping on set $\mathcal{E}$. This diagnostic possesses a degree of stochasticity due to sampling from the set $\mathbf{t}^*_{w_i}$. We can check the impact of this on our impression of the fit by considering a few different random seeds; see the Supplementary Material for further details.

On the whole, the estimated exponential quantiles appear in good agreement with the observed quantiles, indicating the underlying ADF estimate accurately captures the tail behaviour for each min-projection variable. Analogous plots for $\hat{\lambda}_H$ and $\hat{\lambda}_{ST}$ are given in the Supplementary Material. Similar plots were obtained when the other pairs of gauges were considered. 

\begin{figure}[!h]
    \centering
    \includegraphics[width=\textwidth]{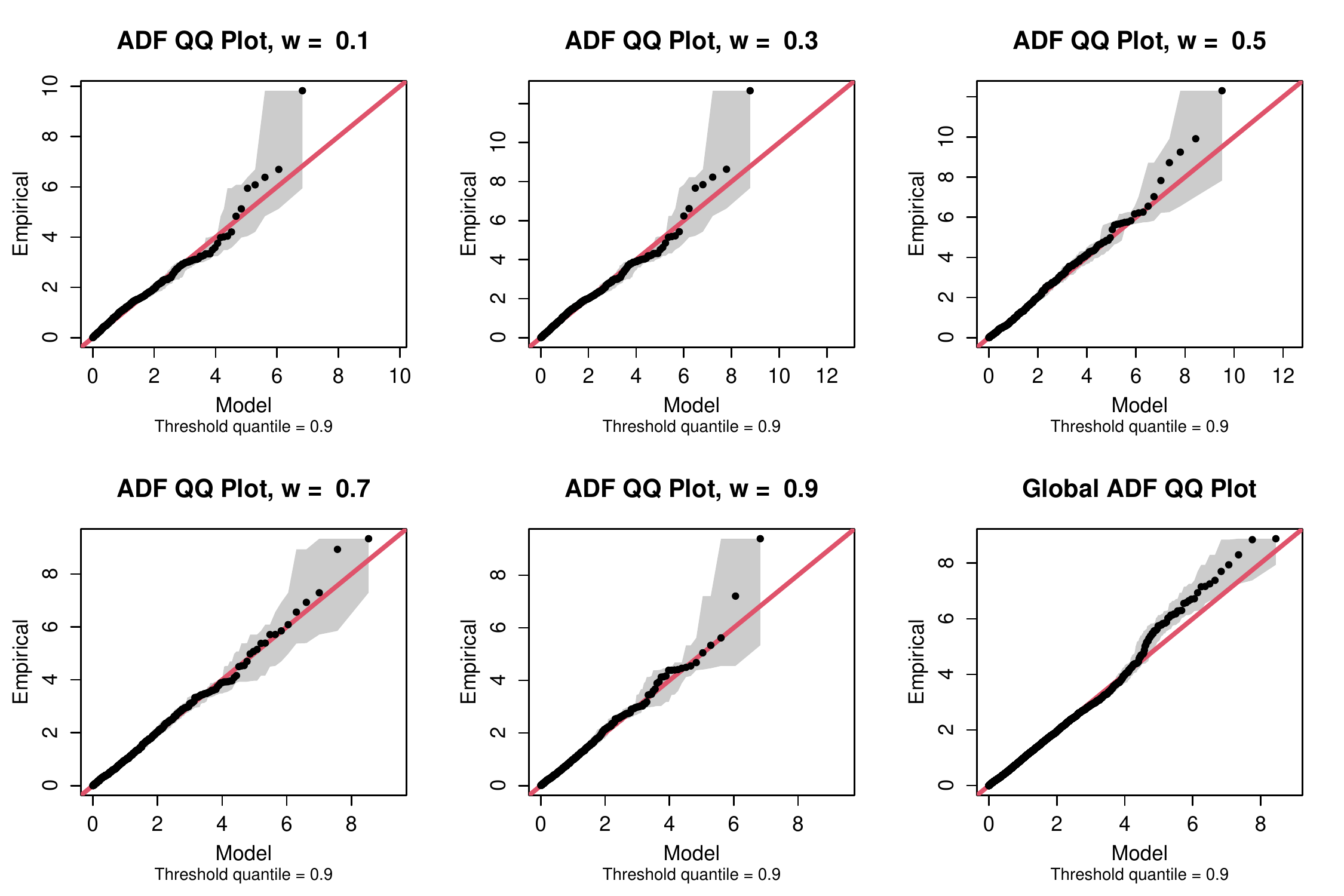}
    \caption{Local and global ADF QQ plots for the third pair of gauges, using the ADF estimate obtained via $\hat{\lambda}_{CL2}$. Estimates given in black, with 95\% pointwise tolerance intervals represented by the grey shaded regions. Red lines correspond to the $y=x$ line.}
    \label{fig:adf_diag}
\end{figure}

\subsection{Estimating return curves} \label{Subsec5.4} 
To quantify the risk of joint flooding events across sites, we follow \citet{Murphy-Barltrop2023} and use the ADF to estimate a bivariate risk measure known as a return curve, $\RC{p}{}$, as defined in Section \ref{Sec1}. This measure is a direct bivariate extension of a return level, which is commonly used to quantify risk in the univariate setting \citep{Coles2001}. Taking probability values $p$ close to zero gives a summary of the joint extremal dependence, thus allowing for comparison across different data structures. In the context of extreme river flows, return curves can be used to evaluate at which sites joint extremes (floods) are more/less likely to occur. For illustration, we fix $p$ to correspond to a $5$ year return period, i.e., $p = 1/(5n_y)$, where $n_y$ is the average number of points observed in a given year \citep{Brunner2016}. Excluding missing observations, we have $28$ years of data, hence the resulting curve should be within the range of data whilst simultaneously representing the joint tail. The resulting return curve estimates for each ADF estimator and pair of gauge sites can be found in Figure \ref{fig:original_return_curves}. 

\begin{figure}[!h]
    \centering
    \includegraphics[width=\textwidth]{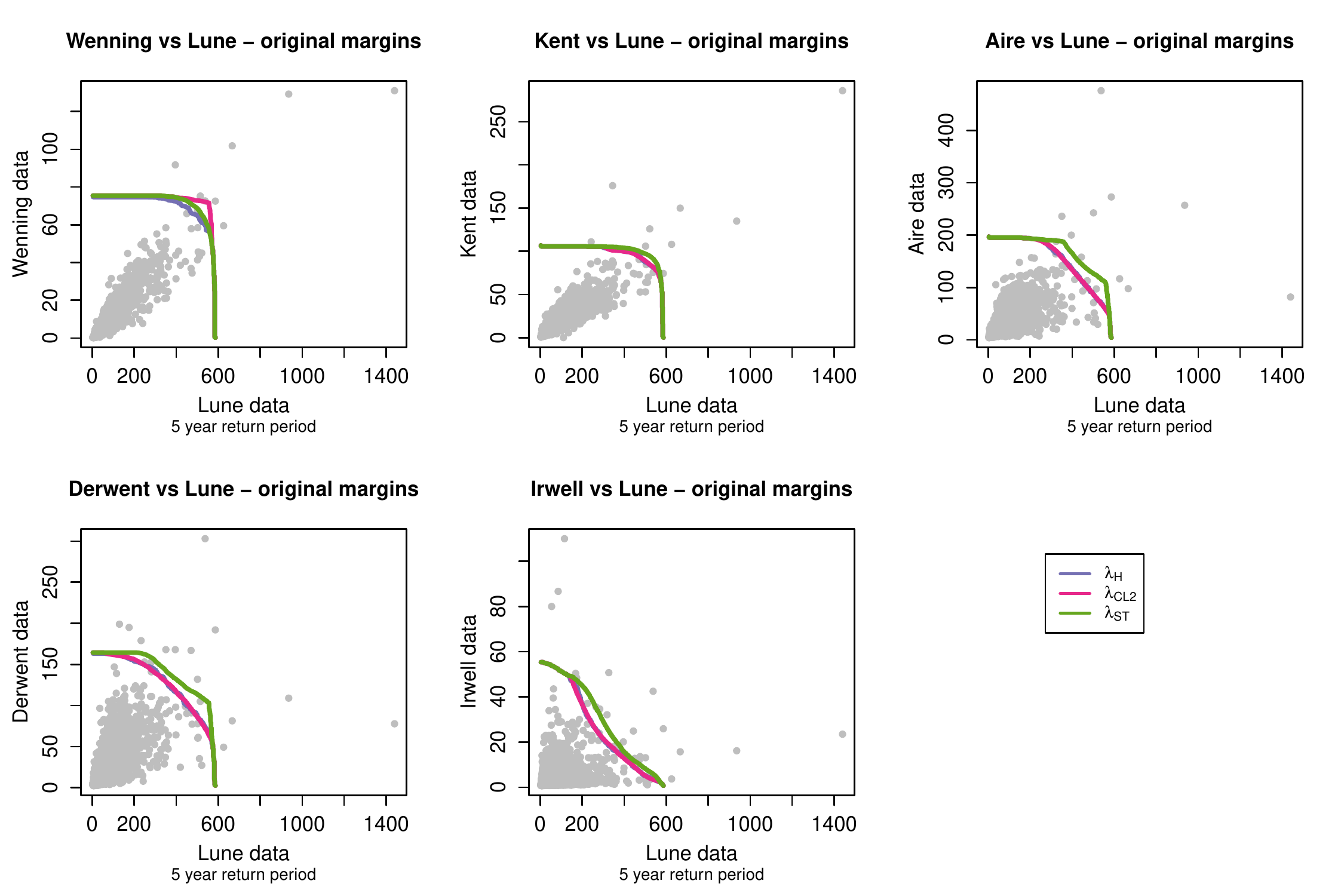}
    \caption{Estimated 5-year return curves (on original margins) for each pair of gauges. The purple, pink and green lines represent the curve estimates from $\hat{\lambda}_{H}$, $\hat{\lambda}_{CL2}$ and $\hat{\lambda}_{ST}$, respectively.}
    \label{fig:original_return_curves}
\end{figure}

There is generally good agreement among the estimated curves. The almost-square shapes of the estimates for the first two pairs of gauges indicate higher likelihoods of observing simultaneous flood events at the corresponding gauge sites; this is as expected given the close spatial proximity of these sites. In all cases, the curves derived via $\hat{\lambda}_H$ are quite rough and unrealistic, and are subsequently ignored. To assess the goodness of fit of the remaining return curve estimates, we consider the first and fifth examples and apply the diagnostic introduced in \citet{Murphy-Barltrop2023}. Our results suggest good quality model fits for both of the estimates obtained using $\hat{\lambda}_{CL2}$ and $\hat{\lambda}_{ST}$, though with potentially a slight preference for the estimate based on $\hat{\lambda}_{CL2}$ for the fifth pair. Furthermore, we also obtain $95\%$ return curve confidence intervals for these examples. The resulting plots illustrating the diagnostics and confidence intervals, along with a brief explanation of the diagnostic tool, are given in the Supplementary Material.

\section{Discussion} \label{Sec6}
We have introduced a range of novel global estimators for the ADF, as detailed in Section \ref{Sec3}. We compared these estimators to existing techniques through a systematic simulation study and found our novel estimators to be competitive in many cases. In particular, the estimators derived via the composite likelihood approach of Section \ref{Subsec3.1}, alongside the estimator of \citet{Simpson2022}, appear to have lower bias and variance, on average, compared to alternative estimation techniques. We also applied ADF estimation techniques to a range of river flow data sets, and obtained estimates of return curves for each data set. The results suggest that our estimation procedures are able to accurately capture the range of extremal dependence structures exhibited in the data.

From Section \ref{Sec4}, one can observe that the `combined' estimators proposed in Section \ref{Subsec3.3} outperform their `uncombined' counterparts in the majority of instances. This indicates that incorporating the knowledge obtained from the conditional extremes parameters leads to improvements in ADF estimates. Furthermore, in most cases, ADF estimates obtained via approximations of the set $\partial G$ appeared to have lower bias compared to alternative estimation techniques. More generally, these results suggest that inferential techniques that exploit the results of \citet{Nolde2022} are superior to techniques which do not. Estimation of $\partial G$, and its impact on estimation of other extremal dependence properties, represents an important line of research.

As noted in Section \ref{Sec1}, few applications of the modelling framework described in equation \eqref{eqn:wads_tawn} exist, even though this model offers advantages over the widely used approach of \citet{Heffernan2004} when evaluating joint tail probabilities. Inference via the ADF ensures consistency in extremal dependence properties, and one can obtain accurate estimates of certain risk measures, such as return curves. 

For each of the existing and novel estimators introduced in Sections \ref{Sec2} and \ref{Sec3}, we were required to select quantile levels, which is equivalent to selecting thresholds of the min-projection. With the exception of $\hat{\lambda}_{ST}$, similar quantile levels were considered for each estimator so as to provide some degree of comparability. However, due to variation in estimation procedures, we acknowledge that the selected quantile levels are not readily comparable since the quantity of joint tail data used for estimation varies across different estimators. Moreover, as noted in Section \ref{Subsec3.4}, trying to select `optimal' quantile levels appears a fruitless exercise since the performance of each estimator does not appear to alter much across different quantile levels.

As noted in Section \ref{Subsec3.4}, our proposed estimators require selection of several tuning parameters. For all practical applications, we recommend this selection is done using a combination of the diagnostic tools outlined in Section \ref{Sec5}, since it is unlikely that one set of tuning parameters will be appropriate across all observed dependence structures and sample sizes.

Finally, we acknowledge the lack of theoretical treatment for the proposed ADF estimators which is an important consideration for understanding properties of the methodology. However, theoretical results of this form typically require in-depth analyses and strict assumptions, which themselves may be hard to verify, whilst in practice one can only ever look at diagnostics obtained from the data. We have therefore opted for a more practical treatment of the proposed estimators.

\section*{Supplementary Material}
\begin{itemize}
    \item \textbf{Supplementary Material for ``Improving estimation for asymptotically independent bivariate extremes via global estimators for the angular dependence functions"}: File containing additional figures and tables. (.pdf file)
    \item \textbf{Code and data}: Zip file containing R scripts and the case study data sets. The script can be used to reproduce results from Sections \ref{Sec4} and \ref{Sec5}. (.zip file)
\end{itemize}

\section*{Declarations}
\subsection*{Funding}
This work was supported by EPSRC grant numbers EP/L015692/1 and EP/X010449/1. 
\subsection*{Competing interests}
The authors have no relevant financial or non-financial interests to disclose.
\subsection*{Data availability}
The data sets analysed in Section \ref{Sec5} are freely available online from the National River Flow Archive \citep{National2022}. The gauge ID numbers for the six considered stations are as follows: 
River Lune - 72004, River Wenning - 72009, River Kent - 73012, River Irwell - 69002, River Aire - 27028, River Derwent - 23007. Date downloaded: October 24th, 2022.  
\subsection*{Ethical Approval}
Not Applicable
\subsection*{Authors' contributions}
All authors contributed to the study conception and design. Material preparation, data collection and analysis were performed by Callum Murphy-Barltrop. The first draft of the manuscript was written by Callum Murphy-Barltrop and all authors commented on previous versions of the manuscript. All authors read and approved the final manuscript.
\subsection*{Acknowledgments}
This paper is based on work partly completed while Callum Murphy-Barltrop was part of the EPSRC funded STOR-i centre for doctoral training (EP/L015692/1). We are grateful to the referee and editor for constructive comments that have improved this article.

\clearpage

\begin{center}
\textbf{\large Supplementary Material to `Improving estimation for asymptotically independent bivariate extremes via global estimators for the angular dependence function'}
\end{center}
\setcounter{equation}{0}
\setcounter{figure}{0}
\setcounter{table}{0}
\setcounter{page}{1}
\setcounter{section}{0}
\makeatletter
\renewcommand{\theequation}{S\arabic{equation}}
\renewcommand{\thefigure}{S\arabic{figure}}
\renewcommand{\thesection}{S\arabic{section}}
\renewcommand{\bibnumfmt}[1]{[S#1]}
\renewcommand{\citenumfont}[1]{S#1}

\section{Post-processing scheme illustration}
Figure \ref{fig:processing_ADF} illustrates an example ADF estimate before and after processing. Observe that enforcing Proposition 3.1 forces the ADF to equal the lower bound for $w \leq 0.39$.

\begin{figure}[h]
    \centering
    \includegraphics[width=.95\textwidth]{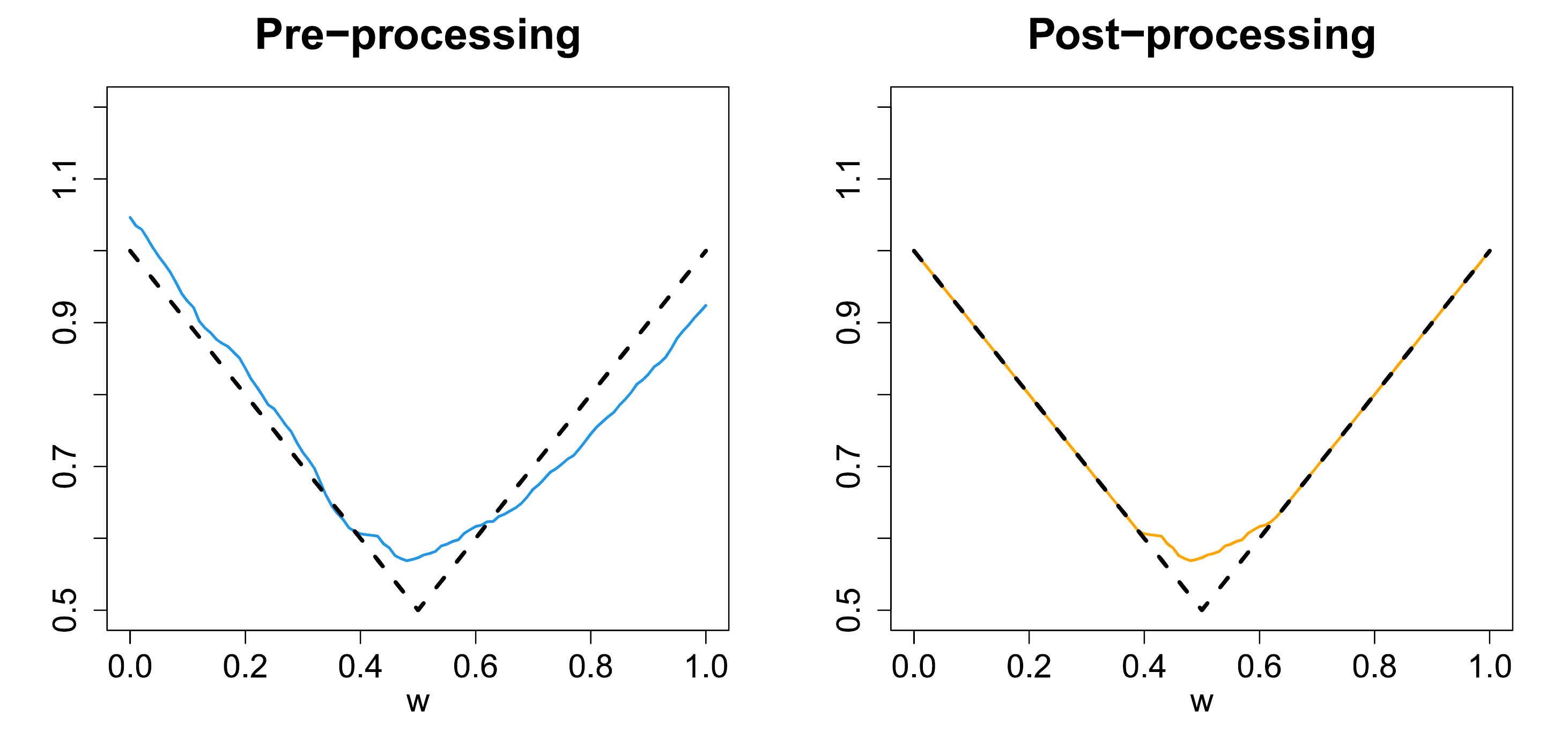}
    \caption{A single ADF estimate before (blue, left) and after (orange, right) processing. The black dotted lines denote the lower bound $\max(w,1-w)$. }
    \label{fig:processing_ADF}
\end{figure}

\section{Example ADF estimates}
Examples of ADF estimates obtained using each of the estimators discussed in the main article are given in Figure \ref{fig:ADF_estimates_examples} for a bivariate Gaussian copula with $\rho = 0.5$. The different estimates are in good agreement, and one can observe the roughness in estimates obtained via the Hill estimator. 
\begin{figure}[!h]
    \centering
    \includegraphics[width=0.6\textwidth]{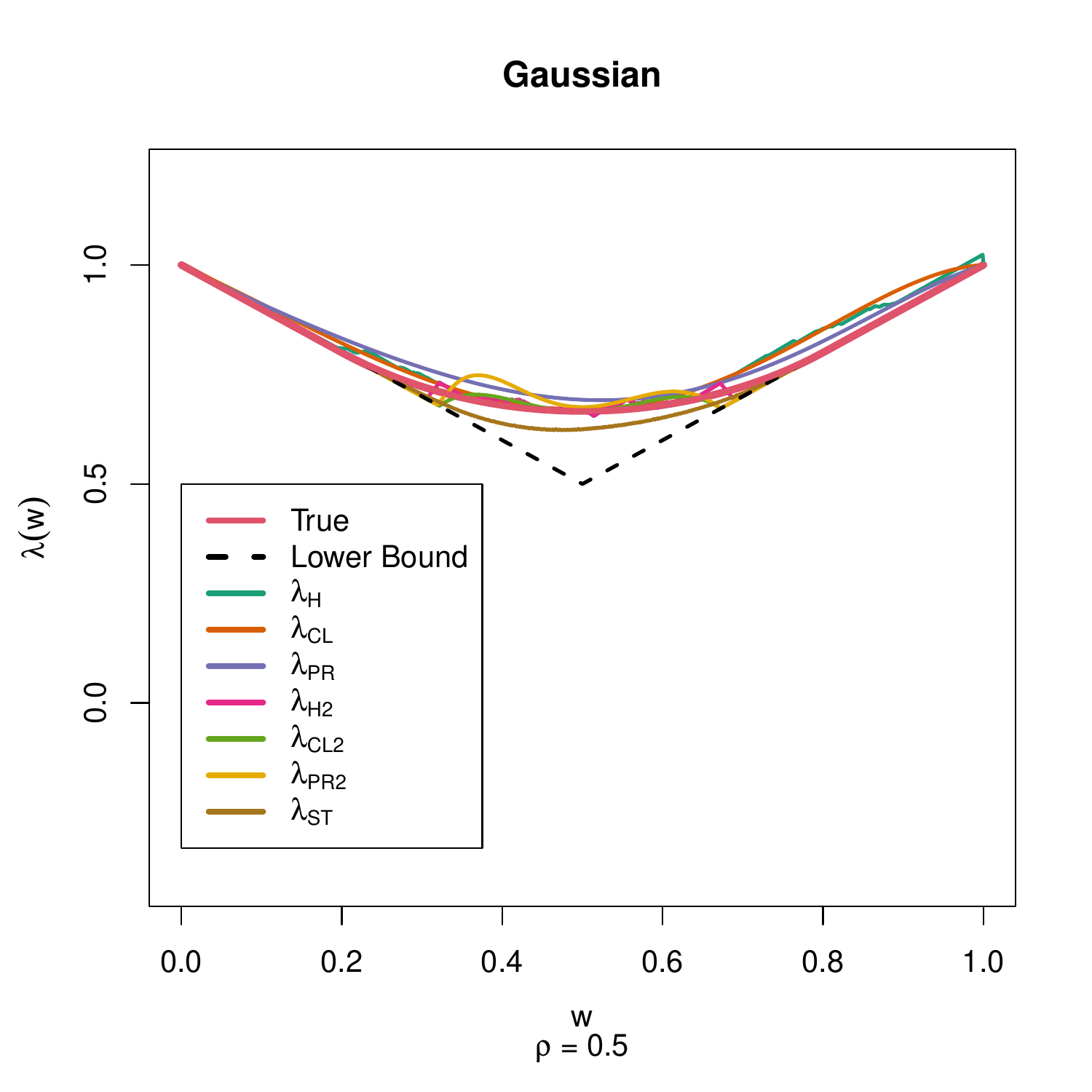}
    \caption{ADF estimates from each of the estimators discussed in the main article. Red represents the true ADF, with the theoretical lower bound given by the dotted black lines.}
    \label{fig:ADF_estimates_examples}
\end{figure}

\section{Example boundary set estimates}

Figure \ref{fig:gauge_estimates} illustrates estimates of the boundary set $\partial G$ obtained using the technique proposed in \citet{Simpson2022} for three copula examples. 


\begin{figure}[!h]
    \centering
    \includegraphics[width=\textwidth]{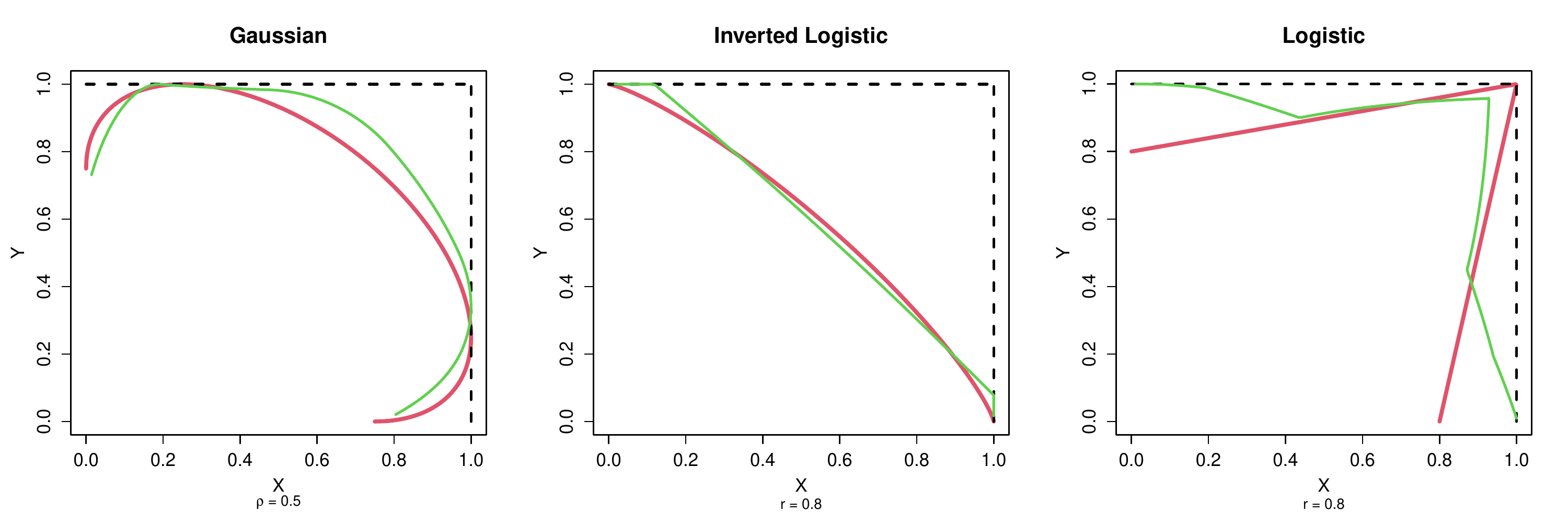}
    \caption{The boundary set $\partial G$ (given in red) for three copula examples, with estimates from \citet{Simpson2022} given in green. Left: bivariate Gaussian copula with correlation coefficient $\rho = 0.5$. Centre: inverted logistic copula with dependence parameter $r = 0.8$. Right: logistic copula with dependence parameter $r = 0.8$. In each plot, the coordinate limits of $\partial G$ are denoted by the black dotted lines.}
    \label{fig:gauge_estimates}
\end{figure}


\section{Tuning parameter selection}
Figures \ref{fig:rmise_est_cl} and \ref{fig:rmise_est_pr} illustrate plots of scaled RMISE estimates against the polynomial degree $k$ for the estimators $\hat{\lambda}_{CL}$ and $\hat{\lambda}_{PR}$, respectively. RMISE was estimated using Monte–Carlo techniques: first, for each $j = 1, 2, \hdots, N$, where $N=200$ denotes the number of samples, the ISE was estimated via the trapezium rule, i.e.,
\begin{equation*}
    \widehat{\text{ISE}}\left(\hat{\lambda}_j\right) = \frac{0.001}{2} \left( (\hat{\lambda}_j(0) - \lambda(0))^2 + \sum_{w \in \mathcal{W}\setminus\{0,1\}}2(\hat{\lambda}_j(w) - \lambda(w))^2 + (\hat{\lambda}_j(1) - \lambda(1))^2 \right),
\end{equation*}
where $\hat{\lambda}_j$ denotes the ADF estimate for sample $j$ and $\mathcal{W}$ denotes the set of rays spanning the interval $[0,1]$, as defined in Section 3.4 the main article. Next, the mean integrated squared error is estimated using Monte–Carlo approximation, i.e., $\widehat{\text{MISE}}:= (1/N)\sum_{j=1}^N\widehat{\text{ISE}}\left(\hat{\lambda}_j \right)$. Finally, an estimate of the RMISE given by
\begin{equation*}
    \widehat{\text{RMISE}} = \sqrt{\widehat{\text{MISE}}}.
\end{equation*}

To quantify simulation error, we also compute the Monte–Carlo variability for each of the $\widehat{\text{RMISE}}$ estimates. We first estimate the Monte–Carlo variance of the ISE estimates, i.e.,
\begin{equation*} 
    \widehat{\text{Var}}(\widehat{\text{ISE}}) = \frac{1}{N-1} \sum_{j=1}^N\left\{\widehat{\text{ISE}}\left(\hat{\lambda}_j \right) - \widehat{\text{MISE}} \right\}^2.
\end{equation*}
Then, since the estimates $\widehat{\text{ISE}}\left(\hat{\lambda}_j \right)$, $j = 1,2,\hdots,N$, are independent and identically distributed, we have
\begin{equation*}
    \widehat{\text{Var}}(\widehat{\text{MISE}}) = \frac{1}{N} \widehat{\text{Var}}(\widehat{\text{ISE}}).
\end{equation*}
Finally, applying the delta method, the Monte–Carlo variance of the RMISE estimate is given by
\begin{equation} \label{eqn:mcv}
    \widehat{\text{Var}}(\widehat{\text{RMISE}}) = \frac{\widehat{\text{Var}}(\widehat{\text{MISE}})}{4\widehat{\text{MISE}}} = \frac{\frac{1}{N-1} \sum_{j=1}^N\left\{\widehat{\text{ISE}}\left(\hat{\lambda}_j \right) - \widehat{\text{MISE}} \right\}^2}{ \frac{4}{N} \sum_{j=1}^N\widehat{\text{ISE}}\left(\hat{\lambda}_j \right)}.
\end{equation}


For the polynomial degree, we considered $k \in \{4,\hdots,11\}$; higher values of $k$ were not considered due to computational complexity. The left and right panels of Figures \ref{fig:rmise_est_cl} and \ref{fig:rmise_est_pr} correspond to Gaussian copulas exhibiting strong ($\rho = 0.9$) and weak ($\rho=0.1$) positive dependence, respectively. For both figures, the error bars correspond to $\widehat{\text{RMISE}} \pm 2\widehat{\text{SD}}(\widehat{\text{RMISE}})$, where $\widehat{\text{SD}}(\widehat{\text{RMISE}}) = \sqrt{\widehat{\text{Var}}(\widehat{\text{RMISE}})}$.

\begin{figure}[!h]
    \centering
    \includegraphics[width=\textwidth]{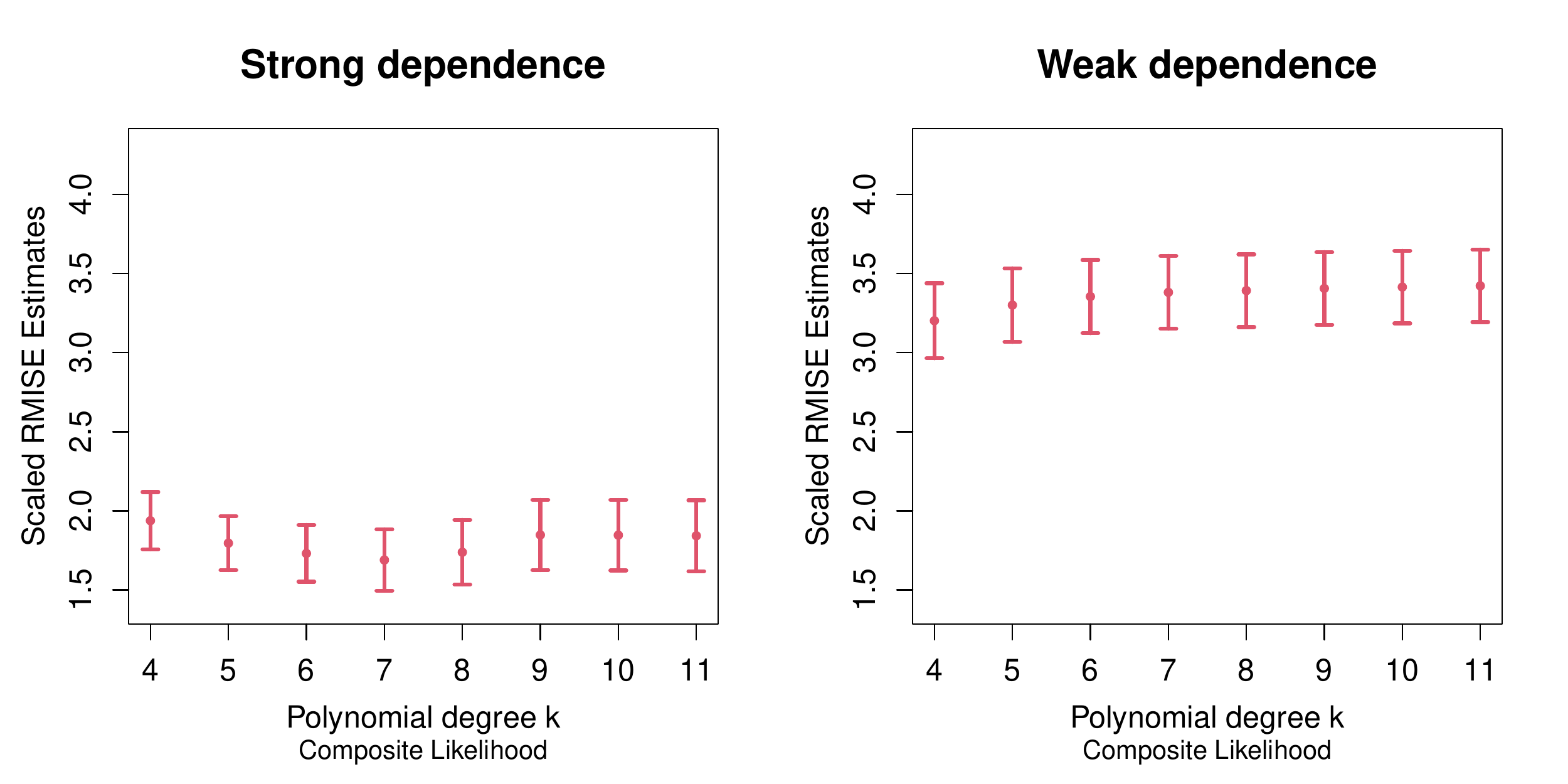}
    \caption{RMISE estimates with error bars (multiplied by 100) over $k$ obtained for $\hat{\lambda}_{CL}$ using $N=200$ from Gaussian copulas with strong (left, $\rho=0.9$) and weak (right, $\rho=0.1$) positive dependence.}
    \label{fig:rmise_est_cl}
\end{figure}

\begin{figure}[!h]
    \centering
    \includegraphics[width=\textwidth]{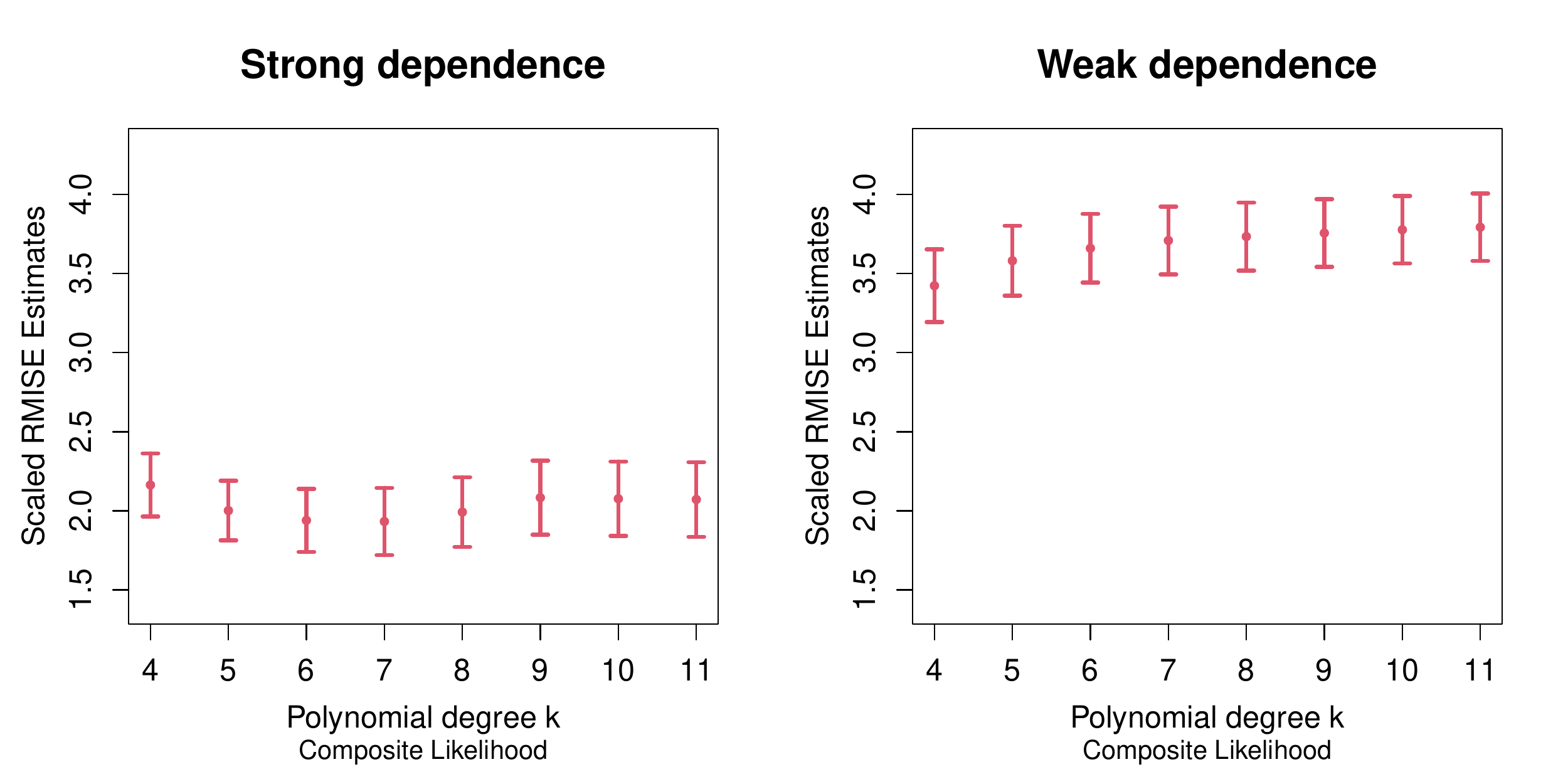}
    \caption{RMISE estimates with error bars (multiplied by 100) over $k$ obtained for $\hat{\lambda}_{PR}$ using $N=200$ from Gaussian copulas with strong (left, $\rho=0.9$) and weak (right, $\rho=0.1$) positive dependence.}
    \label{fig:rmise_est_pr}
\end{figure}

One can observe that, for the strongly dependent example, the RMISE estimates tend to decrease as the value of $k$ increases up to $7$. This is in agreement with findings for the Pickands' dependence function \citep[e.g.,][]{Marcon2017a,Vettori2018}; strongly dependent copulas require a higher degree of flexibility to capture the triangle-like shapes of the resulting dependence functions. However, the RMISE estimates again increase as $k$ increases above $7$, owing to the higher variability that arises for large basis dimensions.


On the other hand, for the weakly dependent Gaussian copula, the value of $k$ made little difference to the resulting RMISE estimates. There is a slight increase in RMISE estimates as $k$ increases; this suggests that having a higher polynomial degrees for such data structures may lead to over-fitting.



In practice, we set $k=7$; this value is sufficient for obtaining low RMISE estimates under both copula examples. Furthermore, for the strongly dependent Gaussian copula, there does not appear to be any advantage in setting $k$ higher than $7$. This polynomial degree therefore appears to offer sufficient flexibility without high computational burden and/or parameter variability.

\section{Additional simulation study results}
The Monte–Carlo errors of the RMISE estimates for each estimator, computed by taking the square root of equation \eqref{eqn:mcv}, are given in Table \ref{table:prec_values}.

\begin{table}[!h]

\caption{Monte–Carlo error of RMISE estimates (multiplied by 100) for each estimator and copula combination, with values reported to 3 significant figures. Copulas 1-3 refer to the Gaussian copula with $\rho = -0.6$, $\rho = 0.1$ and $\rho = 0.6$, respectively. Copulas 4-7 refer to the logistic, asymmetric logistic, inverted logistic, inverted asymmetric logistic copulas, respectively. Copulas 8-9 refer to the student t copula with parameters $\rho = 0.8$, $\nu = 2$ and $\rho = 0.2$, $\nu = 5$, respectively.
\label{table:prec_values}}
\centering
\begin{tabular}[t]{cccccccc}
\toprule
Copula & $\hat{\lambda}_{H}$ & $\hat{\lambda}_{CL}$ & $\hat{\lambda}_{PR}$ & $\hat{\lambda}_{H2}$ & $\hat{\lambda}_{CL2}$ & $\hat{\lambda}_{PR2}$ & $\hat{\lambda}_{ST}$\\
\hline
Copula 1 & 0.076 & 0.0757 & 0.0687 & 0.0788 & 0.0844 & 0.0784 & 0.174\\
\hline
Copula 2 & 0.0431 & 0.0442 & 0.0424 & 0.0427 & 0.0437 & 0.0418 & 0.0469\\
\hline
Copula 3 & 0.0451 & 0.0453 & 0.0462 & 0.0405 & 0.0404 & 0.0402 & 0.0173\\
\hline
Copula 4 & 0.0435 & 0.0444 & 0.0475 & 0.039 & 0.04 & 0.0441 & 0.0624\\
\hline
Copula 5 & 0.0577 & 0.0578 & 0.0566 & 0.0578 & 0.0579 & 0.0563 & 0.0778\\
\hline
Copula 6 & 0.0374 & 0.0376 & 0.04 & 0.031 & 0.0313 & 0.0337 & 0.0373\\
\hline
Copula 7 & 0.0315 & 0.0326 & 0.0323 & 0.0309 & 0.0319 & 0.0312 & 0.0484\\
\hline
Copula 8 & 0.0498 & 0.037 & 0.0411 & 0.0256 & 0.027 & 0.0303 & 0.0685\\
\hline
Copula 9 & 0.0589 & 0.0589 & 0.0578 & 0.0591 & 0.0592 & 0.058 & 0.0938\\
\bottomrule
\end{tabular}
\end{table}

\clearpage

The ISB and IV estimates for each estimator are given in Tables \ref{table:ISB_values} and \ref{table:IV_values}. One can observe that, while the $\hat{\lambda}_{ST}$ appears to perform best in terms of ISB, the $\hat{\lambda}_{CL}$ estimator exhibits the least IV for five out of the nine copula examples. For the most part, one can observe similar ISB and IV values across the different estimators. 

\renewcommand{\arraystretch}{1.5}
\begin{table}[!h]

\caption{ISB values (multiplied by 1,000) for each estimator and copula combination. Smallest ISB values in each row are highlighted in bold, with values reported to 3 significant figures. Copulas 1-3 refer to the Gaussian copula with $\rho = -0.6$, $\rho = 0.1$ and $\rho = 0.6$, respectively. Copulas 4-7 refer to the logistic, asymmetric logistic, inverted logistic, inverted asymmetric logistic copulas, respectively. Copulas 8-9 refer to the student t copula with parameters $\rho = 0.8$, $\nu = 2$ and $\rho = 0.2$, $\nu = 5$, respectively.
\label{table:ISB_values}}
\centering
\begin{tabular}[t]{cccccccc}
\toprule
Copula & $\hat{\lambda}_{H}$ & $\hat{\lambda}_{CL}$ & $\hat{\lambda}_{PR}$ & $\hat{\lambda}_{H2}$ & $\hat{\lambda}_{CL2}$ & $\hat{\lambda}_{PR2}$ & $\hat{\lambda}_{ST}$\\
\hline
Copula 1 & \textbf{372} & 374 & 436 & 374 & 380 & 442 & 402\\
\hline
Copula 2 & 0.348 & 0.366 & 0.452 & 0.321 & 0.333 & 0.409 & \textbf{0.00815}\\
\hline
Copula 3 & 0.812 & 0.847 & 1.04 & 0.753 & 0.77 & 0.943 & \textbf{0.0181}\\
\hline
Copula 4 & 1.74 & 1.84 & 4.2 & 1.52 & 1.49 & 3.25 & \textbf{0.299}\\
\hline
Copula 5 & 19 & 19.2 & 28.4 & 18.9 & 19.1 & 28.1 & \textbf{13.6}\\
\hline
Copula 6 & 0.0385 & 0.0412 & 0.0517 & \textbf{0.00212} & 0.00222 & 0.00365 & 0.0045\\
\hline
Copula 7 & 0.00711 & 0.0041 & 0.00818 & \textbf{0.00157} & 0.0023 & 0.00339 & 0.55\\
\hline
Copula 8 & 0.0236 & 0.0521 & 0.0953 & 0.0107 & \textbf{0.00757} & 0.0157 & 0.0755\\
\hline
Copula 9 & 13.5 & 13.8 & 21.3 & 13.4 & 13.7 & 21.1 & \textbf{10.9}\\
\bottomrule
\end{tabular}
\end{table}

\begin{table}[!h]

\caption{IV values (multiplied by 1,000) for each estimator and copula combination. Smallest IV values in each row are highlighted in bold, with values reported to 3 significant figures. Copulas 1-3 refer to the Gaussian copula with $\rho = -0.6$, $\rho = 0.1$ and $\rho = 0.6$, respectively. Copulas 4-7 refer to the logistic, asymmetric logistic, inverted logistic, inverted asymmetric logistic copulas, respectively. Copulas 8-9 refer to the student t copula with parameters $\rho = 0.8$, $\nu = 2$ and $\rho = 0.2$, $\nu = 5$, respectively.
\label{table:IV_values}}
\centering
\begin{tabular}[t]{cccccccc}
\toprule
Copula & $\hat{\lambda}_{H}$ & $\hat{\lambda}_{CL}$ & $\hat{\lambda}_{PR}$ & $\hat{\lambda}_{H2}$ & $\hat{\lambda}_{CL2}$ & $\hat{\lambda}_{PR2}$ & $\hat{\lambda}_{ST}$\\
\hline
Copula 1 & 2.58 & \textbf{2.02} & 2.11 & 3.09 & 3.23 & 3.13 & 4.7\\
\hline
Copula 2 & 0.834 & \textbf{0.761} & 0.897 & 0.843 & 0.787 & 0.927 & 0.863\\
\hline
Copula 3 & 0.363 & 0.35 & 0.426 & 0.276 & 0.269 & 0.334 & \textbf{0.1}\\
\hline
Copula 4 & 0.362 & 0.371 & 0.551 & \textbf{0.281} & 0.31 & 0.55 & 0.468\\
\hline
Copula 5 & 0.717 & \textbf{0.674} & 0.868 & 0.734 & 0.718 & 0.944 & 1.07\\
\hline
Copula 6 & 0.382 & 0.359 & 0.423 & 0.316 & \textbf{0.306} & 0.366 & 0.445\\
\hline
Copula 7 & 0.772 & \textbf{0.714} & 0.848 & 0.768 & 0.719 & 0.854 & 1.02\\
\hline
Copula 8 & 0.0842 & 0.0586 & 0.113 & \textbf{0.0208} & 0.021 & 0.0362 & 0.273\\
\hline
Copula 9 & 0.751 & \textbf{0.701} & 0.901 & 0.79 & 0.763 & 0.983 & 1.28\\
\bottomrule
\end{tabular}
\end{table}

\clearpage

We also consider the performance of each estimator at different rays (i.e., different regions in $\RR^2_+$). For this, estimates of the the root mean squared error (RMSE) were computed via Monte–Carlo techniques for the rays $w \in \{0.1,0.3,0.5,0.7,0.9\}$; these estimates are given in Tables \ref{table:RMSE_values0.1}-\ref{table:RMSE_values0.9}.

These results do not reveal any obvious patterns, besides the fact the $\hat{\lambda}_{ST}$ estimator tends to perform best for $w = 0.5$ (i.e., on the $y=x$ line). As with the RMISE, ISB and IV estimates, the performance of each estimator appears to vary over copulas. This is likely due to the different rates of convergence to the limiting ADF for finite sample sizes.

We also observe that for $w=0.1$ and $w=0.9$, the $\hat{\lambda}_{ST}$ has some notably small RMSE estimates for the Gaussian copula with $\rho = 0.6$. What this shows is that for this dependence structure, the estimates of $\partial G$ at each iteration touch the boundary of the unit box $[0,1]^2$ at some ray in the ranges $[0.1,0.5]$ and $[0.5,0.9]$, ensuring that the estimated ADF correctly equals the lower bound for $w = 0.1, 0.9$. The estimators $\hat{\lambda}_{H2}$, $\hat{\lambda}_{CL2}$ and $\hat{\lambda}_{PR2}$ will all lead to this conclusion if the estimates of the conditional extremes parameters satisfy $\alpha^*_{x \mid y} > 0.1$ and $\alpha^*_{y \mid x} < 0.9$. This explains the notable improvements over their non-combined counterparts in this case, and similar improvements for copulas $4$ and $8$.


\begin{table}[!h]

\caption{RMSE values (multiplied by 100) for each estimator and copula combination at $w = 0.1$. Smallest RMSE values in each row are highlighted in bold, with values reported to 3 significant figures. Copulas 1-3 refer to the Gaussian copula with $\rho = -0.6$, $\rho = 0.1$ and $\rho = 0.6$, respectively. Copulas 4-7 refer to the logistic, asymmetric logistic, inverted logistic, inverted asymmetric logistic copulas, respectively. Copulas 8-9 refer to the student t copula with parameters $\rho = 0.8$, $\nu = 2$ and $\rho = 0.2$, $\nu = 5$, respectively.
\label{table:RMSE_values0.1}}
\centering
\begin{tabular}[t]{cccccccc}
\toprule
Copula & $\hat{\lambda}_{H}$ & $\hat{\lambda}_{CL}$ & $\hat{\lambda}_{PR}$ & $\hat{\lambda}_{H2}$ & $\hat{\lambda}_{CL2}$ & $\hat{\lambda}_{PR2}$ & $\hat{\lambda}_{ST}$\\
\hline
Copula 1 & 59.7 & \textbf{56.9} & 62.2 & 59.8 & 58.3 & 63.5 & 61.8\\
\hline
Copula 2 & 3.45 & 3.74 & 4.12 & 3.48 & 3.7 & 4.02 & \textbf{2.47}\\
\hline
Copula 3 & 1.76 & 2.15 & 2.4 & 0.424 & 0.416 & 0.436 & $\mathbf{1.11 \times 10^{-14}}$\\
\hline
Copula 4 & 2.01 & 2.52 & 3.19 & 0.424 & 0.458 & \textbf{0.42} & 2.78\\
\hline
Copula 5 & 4.15 & 4.68 & 5.48 & 3.27 & \textbf{3.1} & 3.54 & 4.58\\
\hline
Copula 6 & 1.88 & 2.21 & 2.42 & 1.08 & 1.09 & 1.17 & \textbf{0.175}\\
\hline
Copula 7 & \textbf{3} & 3.29 & 3.69 & 3.15 & 3.44 & 3.82 & 4.84\\
\hline
Copula 8 & 1.07 & 0.755 & 1.27 & \textbf{0.237} & 0.244 & 0.315 & 2.42\\
\hline
Copula 9 & 6.76 & 7.59 & 8.72 & 6.64 & 7.27 & 8.26 & \textbf{5.98}\\
\bottomrule
\end{tabular}
\end{table}

\begin{table}[!h]

\caption{RMSE values (multiplied by 100) for each estimator and copula combination at $w = 0.3$. Smallest RMSE values in each row are highlighted in bold, with values reported to 3 significant figures. Copulas 1-3 refer to the Gaussian copula with $\rho = -0.6$, $\rho = 0.1$ and $\rho = 0.6$, respectively. Copulas 4-7 refer to the logistic, asymmetric logistic, inverted logistic, inverted asymmetric logistic copulas, respectively. Copulas 8-9 refer to the student t copula with parameters $\rho = 0.8$, $\nu = 2$ and $\rho = 0.2$, $\nu = 5$, respectively.
\label{table:RMSE_values0.3}}
\centering
\begin{tabular}[t]{cccccccc}
\toprule
Copula & $\hat{\lambda}_{H}$ & $\hat{\lambda}_{CL}$ & $\hat{\lambda}_{PR}$ & $\hat{\lambda}_{H2}$ & $\hat{\lambda}_{CL2}$ & $\hat{\lambda}_{PR2}$ & $\hat{\lambda}_{ST}$\\
\hline
Copula 1 & \textbf{62} & 63.7 & 68.6 & 62 & 63.7 & 68.6 & 63.8\\
\hline
Copula 2 & 3.5 & 3.36 & 3.63 & 3.5 & 3.31 & 3.59 & \textbf{3.13}\\
\hline
Copula 3 & 3.65 & 3.67 & 4.07 & 3.64 & 3.85 & 4.24 & \textbf{0.71}\\
\hline
Copula 4 & 3.22 & 3.2 & 5.31 & \textbf{2.59} & 2.6 & 3.58 & 2.93\\
\hline
Copula 5 & 10.2 & 10.1 & 13.7 & 10.2 & 10.3 & 13.8 & \textbf{9.83}\\
\hline
Copula 6 & 2.18 & \textbf{2.01} & 2.22 & 2.19 & 2.19 & 2.44 & 2.67\\
\hline
Copula 7 & 2.89 & \textbf{2.66} & 2.91 & 2.89 & 2.68 & 2.95 & 4.46\\
\hline
Copula 8 & 0.837 & 0.587 & 0.99 & \textbf{0.257} & 0.26 & 0.337 & 1.88\\
\hline
Copula 9 & 10.9 & \textbf{10.7} & 13.9 & 10.9 & 10.8 & 13.9 & 11.6\\
\bottomrule
\end{tabular}
\end{table}

\begin{table}[!h]

\caption{RMSE values (multiplied by 100) for each estimator and copula combination at $w = 0.5$. Smallest RMSE values in each row are highlighted in bold, with values reported to 3 significant figures. Copulas 1-3 refer to the Gaussian copula with $\rho = -0.6$, $\rho = 0.1$ and $\rho = 0.6$, respectively. Copulas 4-7 refer to the logistic, asymmetric logistic, inverted logistic, inverted asymmetric logistic copulas, respectively. Copulas 8-9 refer to the student t copula with parameters $\rho = 0.8$, $\nu = 2$ and $\rho = 0.2$, $\nu = 5$, respectively.
\label{table:RMSE_values0.5}}
\centering
\begin{tabular}[t]{cccccccc}
\hline
Copula & $\hat{\lambda}_{H}$ & $\hat{\lambda}_{CL}$ & $\hat{\lambda}_{PR}$ & $\hat{\lambda}_{H2}$ & $\hat{\lambda}_{CL2}$ & $\hat{\lambda}_{PR2}$ & $\hat{\lambda}_{ST}$\\
\toprule
Copula 1 & 62.6 & \textbf{60.8} & 66.6 & 62.6 & \textbf{60.8} & 66.6 & 63.9\\
\hline
Copula 2 & 3.65 & 3.49 & 3.75 & 3.65 & 3.59 & 3.85 & \textbf{3.49}\\
\hline
Copula 3 & 5.13 & 5.08 & 5.56 & 5.13 & 5.18 & 5.74 & \textbf{2.39}\\
\hline
Copula 4 & 12.1 & 12.4 & 16.3 & 12.1 & 12 & 16.2 & \textbf{4.63}\\
\hline
Copula 5 & 26.4 & 26.6 & 32 & 26.4 & 26.6 & 32.1 & \textbf{20.8}\\
\hline
Copula 6 & 2.03 & \textbf{1.96} & 2.05 & 2.03 & 1.97 & 2.15 & 2.67\\
\hline
Copula 7 & 2.8 & \textbf{2.62} & 2.73 & 2.8 & 2.63 & 2.78 & 3.61\\
\hline
Copula 8 & 2.79 & 4.27 & 5 & 2.78 & 2.52 & 3.26 & \textbf{1.8}\\
\hline
Copula 9 & 23.3 & 23.5 & 27.9 & 23.3 & 23.5 & 27.9 & \textbf{19}\\
\bottomrule
\end{tabular}
\end{table}

\begin{table}[!h]

\caption{RMSE values (multiplied by 100) for each estimator and copula combination at $w = 0.7$. Smallest RMSE values in each row are highlighted in bold, with values reported to 3 significant figures. Copulas 1-3 refer to the Gaussian copula with $\rho = -0.6$, $\rho = 0.1$ and $\rho = 0.6$, respectively. Copulas 4-7 refer to the logistic, asymmetric logistic, inverted logistic, inverted asymmetric logistic copulas, respectively. Copulas 8-9 refer to the student t copula with parameters $\rho = 0.8$, $\nu = 2$ and $\rho = 0.2$, $\nu = 5$, respectively.
\label{table:RMSE_values0.7}}
\centering
\begin{tabular}[t]{cccccccc}
\toprule
Copula & $\hat{\lambda}_{H}$ & $\hat{\lambda}_{CL}$ & $\hat{\lambda}_{PR}$ & $\hat{\lambda}_{H2}$ & $\hat{\lambda}_{CL2}$ & $\hat{\lambda}_{PR2}$ & $\hat{\lambda}_{ST}$\\
\hline
Copula 1 & \textbf{62.1} & 63.8 & 68.4 & 62.1 & 63.8 & 68.4 & 63.8\\
\hline
Copula 2 & 3.53 & 3.38 & 3.8 & 3.53 & 3.34 & 3.76 & \textbf{3.27}\\
\hline
Copula 3 & 3.7 & 3.67 & 4.13 & 3.7 & 3.89 & 4.35 & \textbf{0.592}\\
\hline
Copula 4 & 3.1 & 3.08 & 5.16 & \textbf{2.52} & 2.6 & 3.69 & 2.82\\
\hline
Copula 5 & 16.8 & 16.9 & 19.3 & 16.8 & 16.9 & 19.2 & \textbf{16.4}\\
\hline
Copula 6 & 2.21 & \textbf{2.06} & 2.27 & 2.22 & 2.25 & 2.49 & 2.71\\
\hline
Copula 7 & 2.78 & \textbf{2.54} & 2.66 & 2.78 & 2.57 & 2.71 & 4.42\\
\hline
Copula 8 & 0.835 & 0.571 & 0.858 & 0.233 & \textbf{0.227} & 0.269 & 1.85\\
\hline
Copula 9 & 10.9 & \textbf{10.7} & 13.9 & 10.9 & 10.7 & 13.9 & 11.6\\
\bottomrule
\end{tabular}
\end{table}

\begin{table}[!h]

\caption{RMSE values (multiplied by 100) for each estimator and copula combination at $w = 0.9$. Smallest RMSE values in each row are highlighted in bold, with values reported to 3 significant figures. Copulas 1-3 refer to the Gaussian copula with $\rho = -0.6$, $\rho = 0.1$ and $\rho = 0.6$, respectively. Copulas 4-7 refer to the logistic, asymmetric logistic, inverted logistic, inverted asymmetric logistic copulas, respectively. Copulas 8-9 refer to the student t copula with parameters $\rho = 0.8$, $\nu = 2$ and $\rho = 0.2$, $\nu = 5$, respectively.
\label{table:RMSE_values0.9}}
\centering
\begin{tabular}[t]{cccccccc}
\toprule
Copula & $\hat{\lambda}_{H}$ & $\hat{\lambda}_{CL}$ & $\hat{\lambda}_{PR}$ & $\hat{\lambda}_{H2}$ & $\hat{\lambda}_{CL2}$ & $\hat{\lambda}_{PR2}$ & $\hat{\lambda}_{ST}$\\
\hline
Copula 1 & 59.7 & \textbf{57.1} & 62.5 & 59.8 & 58.4 & 63.7 & 61.7\\
\hline
Copula 2 & 3.38 & 3.65 & 3.93 & 3.44 & 3.62 & 3.85 & \textbf{2.54}\\
\hline
Copula 3 & 1.84 & 2.24 & 2.47 & 0.352 & 0.289 & 0.268 & $\mathbf{1.11 \times {10}^{-14}}$\\
\hline
Copula 4 & 2 & 2.5 & 3.25 & \textbf{0.392} & 0.394 & 0.431 & 2.71\\
\hline
Copula 5 & 7.22 & 7.42 & 7.91 & 7.21 & 7.48 & 7.94 & \textbf{2.58}\\
\hline
Copula 6 & 2.05 & 2.37 & 2.56 & 1.04 & 1.06 & 1.16 & \textbf{0.183}\\
\hline
Copula 7 & \textbf{2.76} & 3.03 & 3.36 & 2.82 & 2.86 & 3.07 & 2.8\\
\hline
Copula 8 & 1.06 & 0.734 & 1.1 & \textbf{0.143} & 0.157 & 0.183 & 2.38\\
\hline
Copula 9 & 6.6 & 7.41 & 8.4 & 6.51 & 7.12 & 7.98 & \textbf{5.91}\\
\bottomrule
\end{tabular}
\end{table}

\clearpage

\section{Additional case study figures}

This section contains additional figures for the case study detailed in Section 4 of the main article. Figure \ref{fig:river_flow_time_series} illustrates daily river flow time series for each of the six gauges in the north of England, UK. These series suggest a stationarity assumption is reasonable for the extremes of each data set. 

\begin{figure}[!h]
    \centering
    \includegraphics[width=\textwidth]{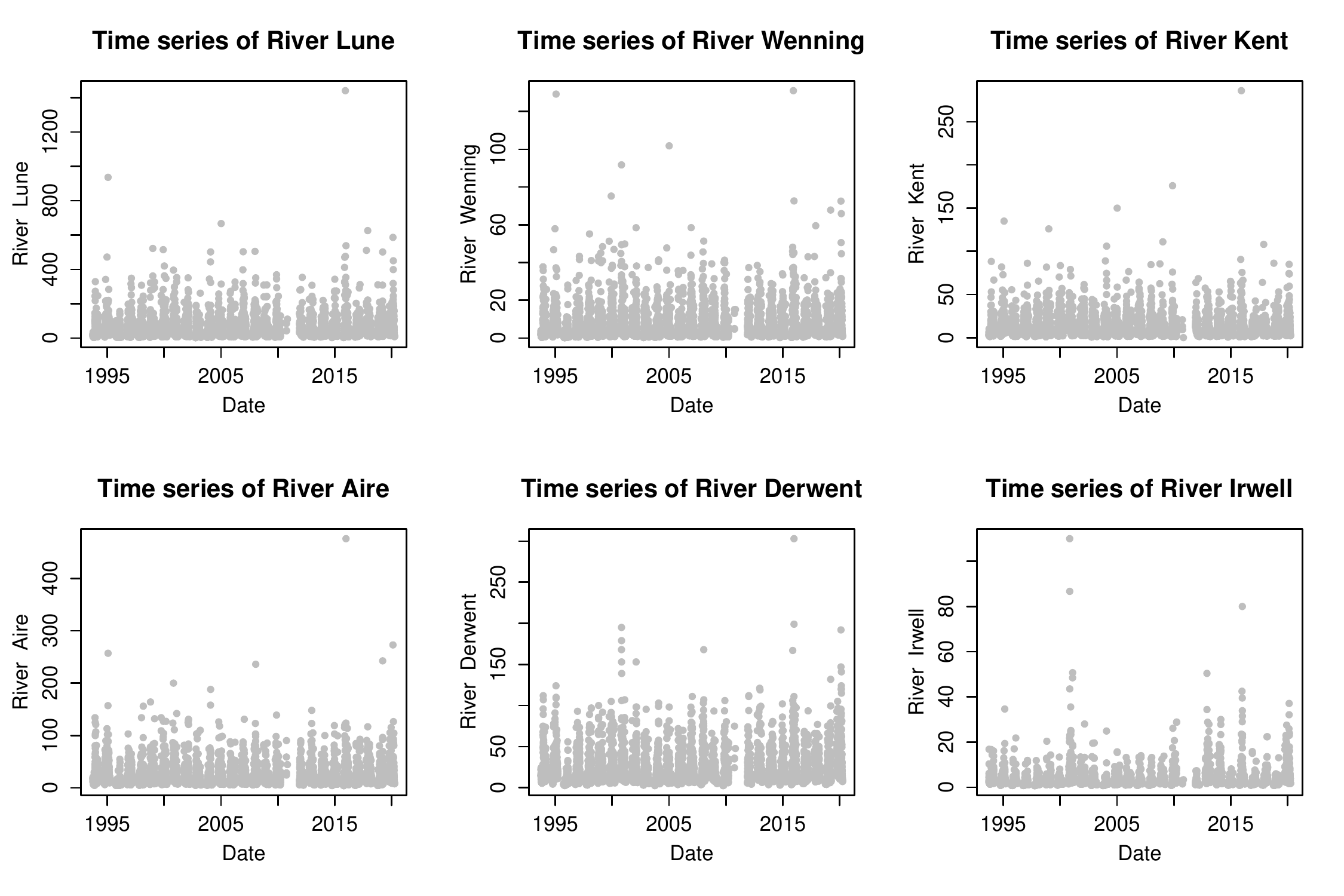}
    \caption{Daily river flow time series for the six gauges in the north of England, UK.}
    \label{fig:river_flow_time_series}
\end{figure}

Figure \ref{fig:qq_plots} illustrates the QQ plots from the fitted GPDs at each of the six gauges. Uncertainty intervals are obtained via block bootstrapping on the order statistics; see the main article for further details. One can observe that, in each case, the majority of points lie close to the $y=x$ line, indicating the fitted models capture the upper tails well. 

\begin{figure}[!h]
    \centering
    \includegraphics[width=\textwidth]{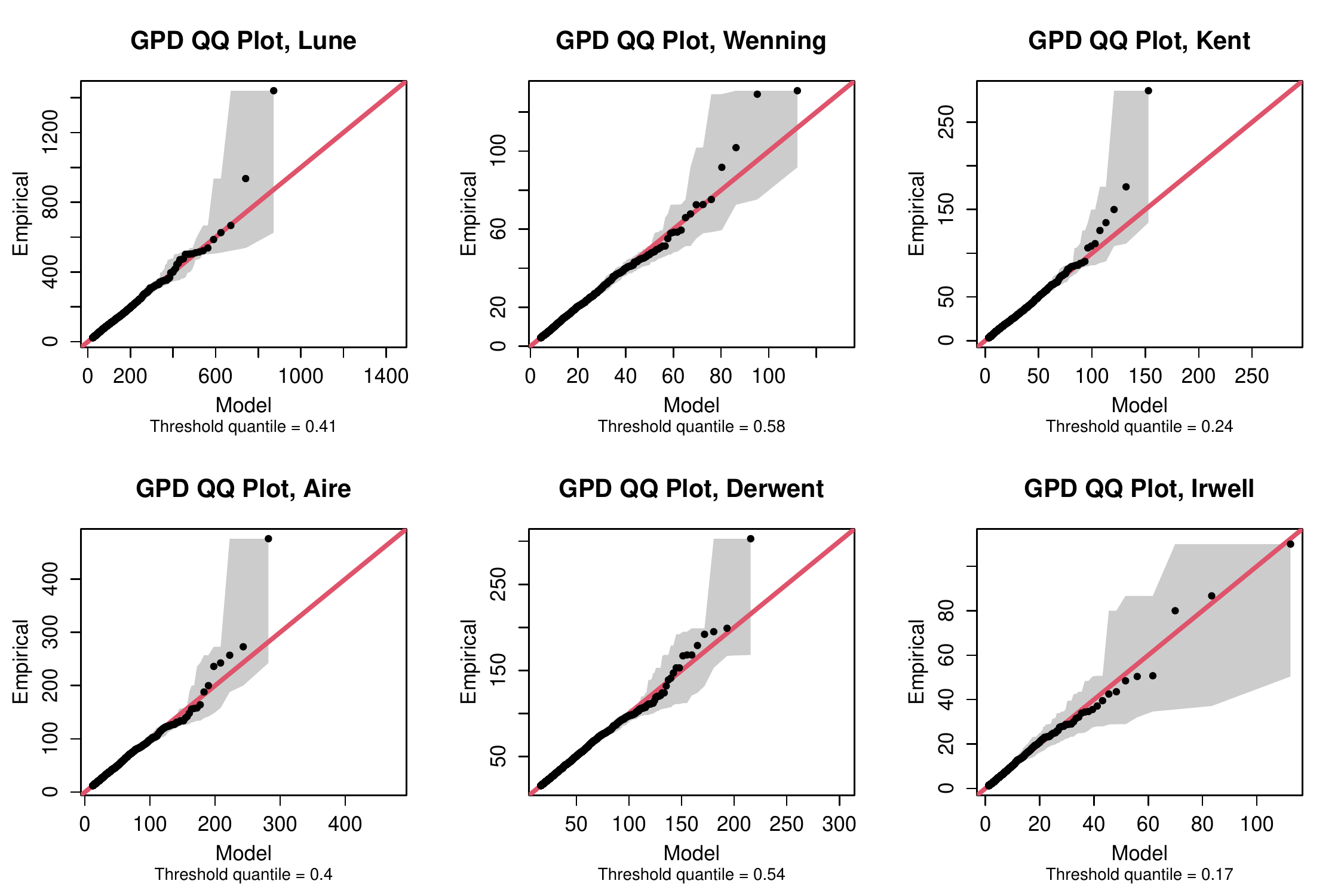}
    \caption{QQ plots for each of the fitted GPDs at each of the six gauges. Estimates are given in black, with 95\% bootstrapped tolerance intervals represented by the grey shaded regions. The red line corresponds to the $y=x$ line. The corresponding threshold quantile levels are given in the subtitle of each plot.}
    \label{fig:qq_plots}
\end{figure}

Figures \ref{fig:adf_diag2} and \ref{fig:adf_diag3} illustrate the ADF QQ plots for the third pair of gauges using the estimates obtained via $\hat{\lambda}_{ST}$ and $\hat{\lambda}_{H}$, respectively. The estimated and observed quantiles appear in good agreement at each of the considered rays. 


\begin{figure}[!h]
    \centering
    \includegraphics[width=\textwidth]{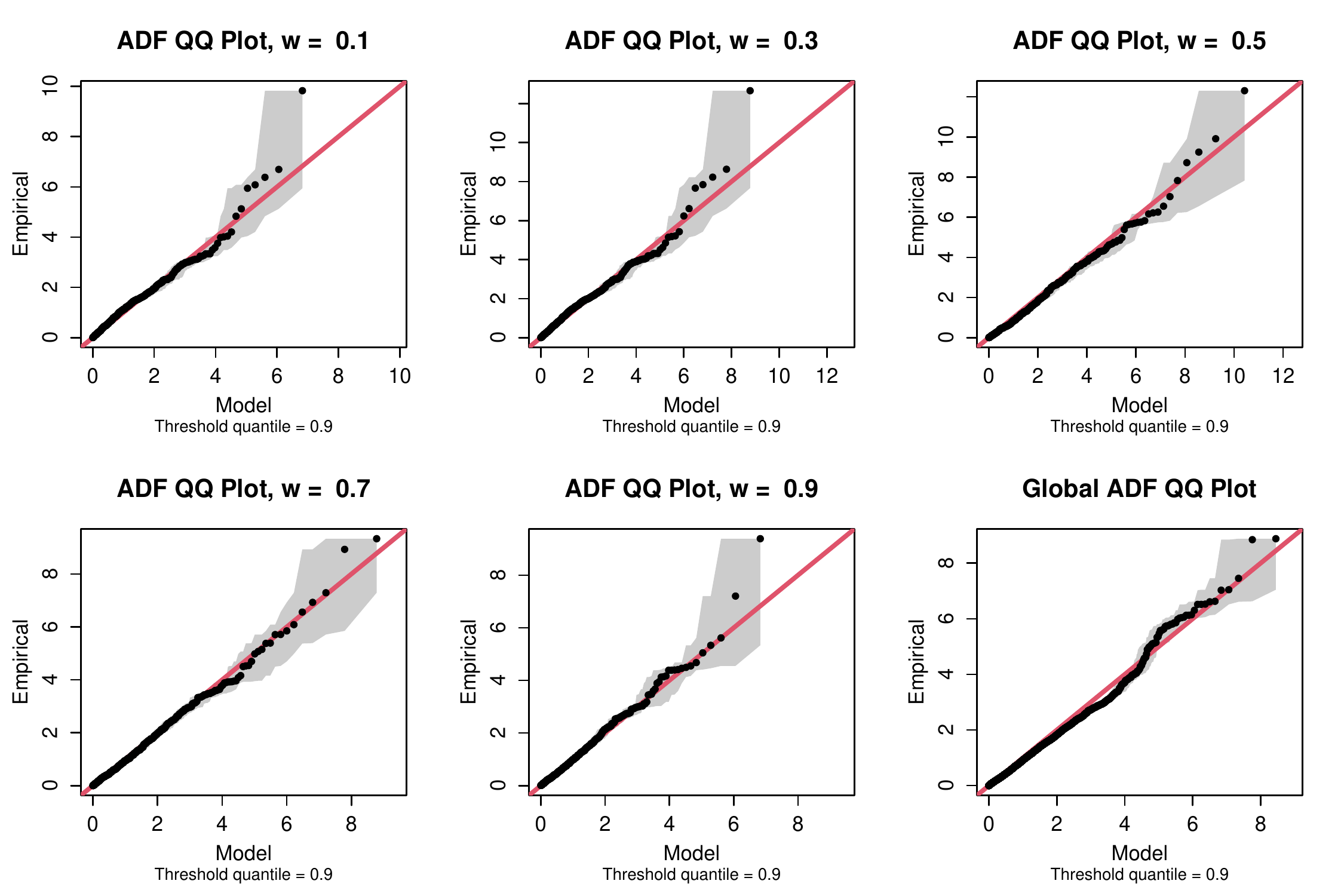}
    \caption{Local and global ADF QQ plots for the third pair of gauges, using the ADF estimate obtained via $\hat{\lambda}_{ST}$. Estimates given in black, with 95\% pointwise tolerance intervals represented by the grey shaded regions. Red lines correspond to the $y=x$ line.}
    \label{fig:adf_diag2}
\end{figure}

\begin{figure}[!h]
    \centering
    \includegraphics[width=\textwidth]{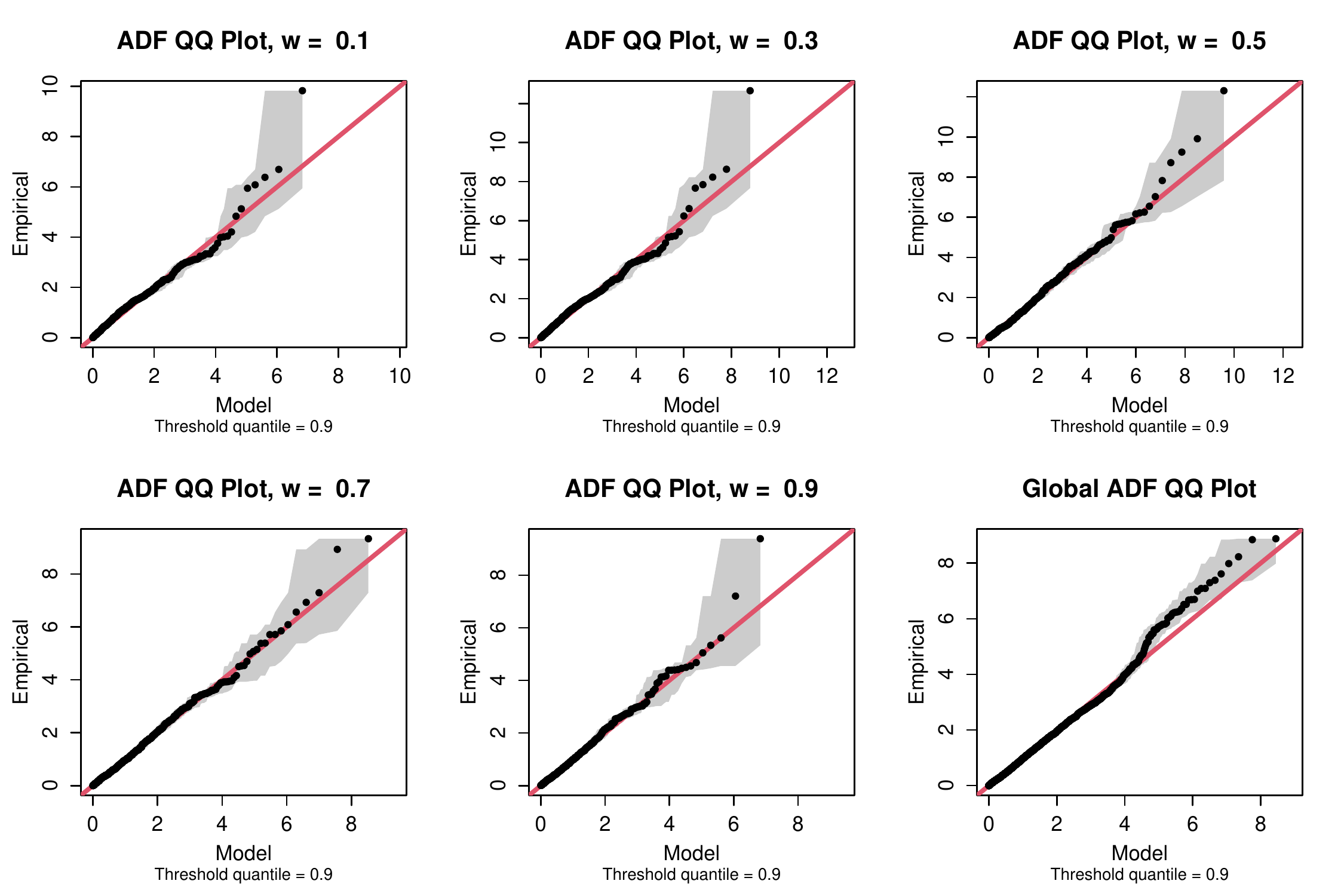}
    \caption{Local and global ADF QQ plots for the third pair of gauges, using the ADF estimate obtained via $\hat{\lambda}_{H}$. Estimates given in black, with 95\% pointwise tolerance intervals represented by the grey shaded regions. Red lines correspond to the $y=x$ line.}
    \label{fig:adf_diag3}
\end{figure}

Figures \ref{fig:adf_global}, \ref{fig:adf_global2} and \ref{fig:adf_global3} illustrate the global ADF diagnostic over six random seeds for the third pair of gauges using the estimates obtained via $\hat{\lambda}_{CL2}$, $\hat{\lambda}_{ST}$ and $\hat{\lambda}_{H}$, respectively. One can clearly observe the variability in this diagnostic over different random seeds. Therefore, in practice, we recommend computing this diagnostic multiple times to ensure one can fairly compare the model fits from each estimator.

\begin{figure}[!h]
    \centering
    \includegraphics[width=\textwidth]{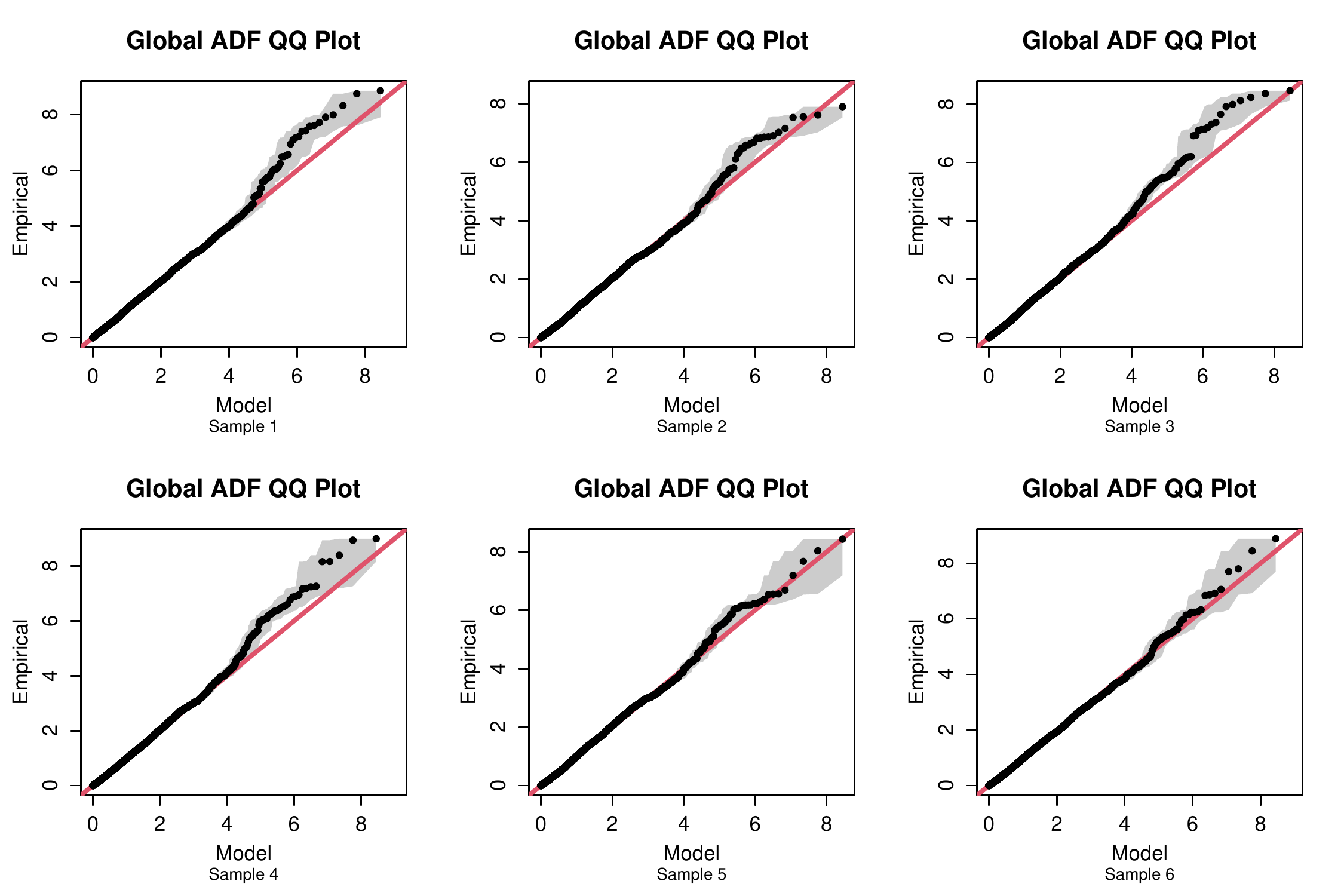}
    \caption{Global ADF QQ plot over six random seeds for the third pair of gauges, using the ADF estimate obtained via $\hat{\lambda}_{CL2}$. Estimates given in black, with 95\% pointwise tolerance intervals represented by the grey shaded regions. Red lines correspond to the $y=x$ line.}
    \label{fig:adf_global}
\end{figure}

\begin{figure}[!h]
    \centering
    \includegraphics[width=\textwidth]{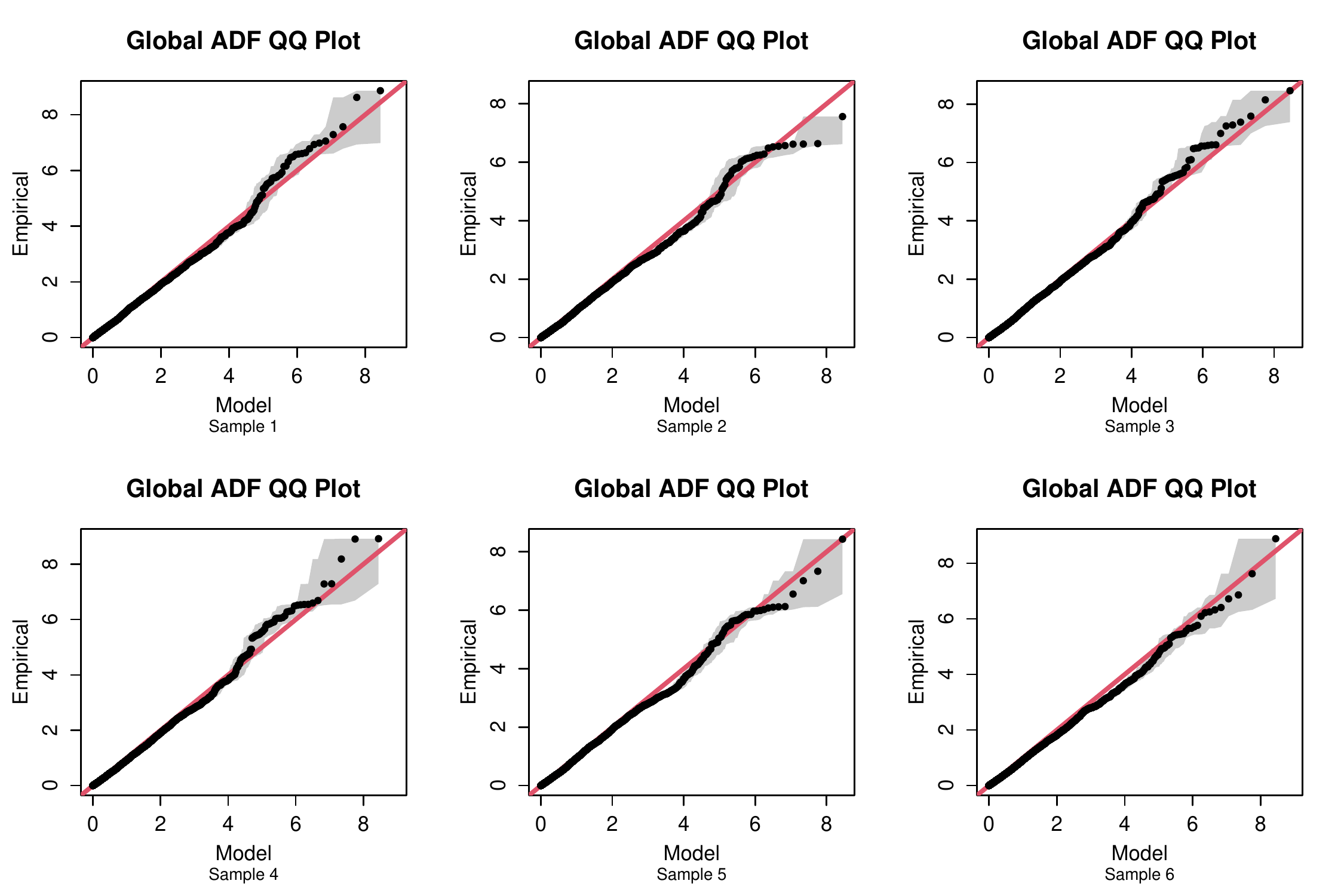}
    \caption{Global ADF QQ plot over six random seeds for the third pair of gauges, using the ADF estimate obtained via $\hat{\lambda}_{ST}$. Estimates given in black, with 95\% pointwise tolerance intervals represented by the grey shaded regions. Red lines correspond to the $y=x$ line.}
    \label{fig:adf_global2}
\end{figure}

\begin{figure}[!h]
    \centering
    \includegraphics[width=\textwidth]{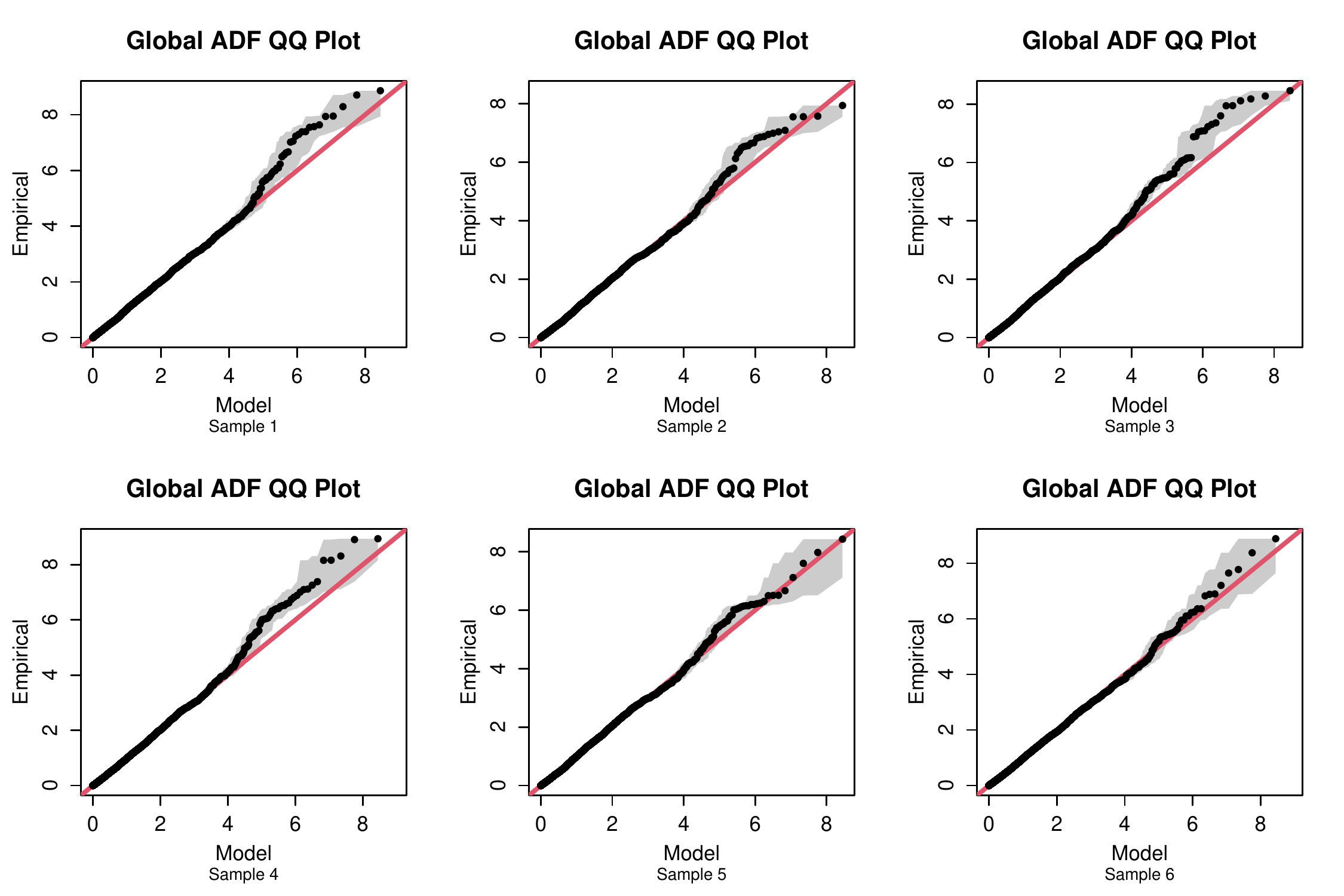}
    \caption{Global ADF QQ plot over six random seeds for the third pair of gauges, using the ADF estimate obtained via $\hat{\lambda}_{H}$. Estimates given in black, with 95\% pointwise tolerance intervals represented by the grey shaded regions. Red lines correspond to the $y=x$ line.}
    \label{fig:adf_global3}
\end{figure}

Figures \ref{fig:diag_st} and \ref{fig:diag_cl2} illustrate the return curve diagnostic of \citet{Murphy-Barltrop2023} for the estimators $\hat{\lambda}_{ST}$ and $\hat{\lambda}_{CL2}$, respectively. Note that for both of these figures, the `true survival probabilities' refer to the probability $p$ at which return curve estimates have been obtained. For this diagnostic, a subset of points are selected on a return curve estimate; these points correspond to a set of $m=150$ equally spaced angles $\theta$ in the interval $[0,\pi/2]$, i.e., given $(x,y) \in \RC{p}{}$, we have $\theta = \tan^{-1}(y/x)$. Empirical estimates of the joint survival function are computed for each point and bootstrapping is used to evaluate uncertainty. Finally, the median empirical estimates, alongside 95\% pointwise confidence intervals, are plotted against the angle index and compared to the true probability; see \citet{Murphy-Barltrop2023} for further details. 

Both estimators appear to give a similar level of accuracy, though for the fifth gauge site pairing, both $\hat{\lambda}_{ST}$ and $\hat{\lambda}_{CL2}$ fail to capture the true probability at all angles. 

\begin{figure}[!h]
    \centering
    \includegraphics[width=\textwidth]{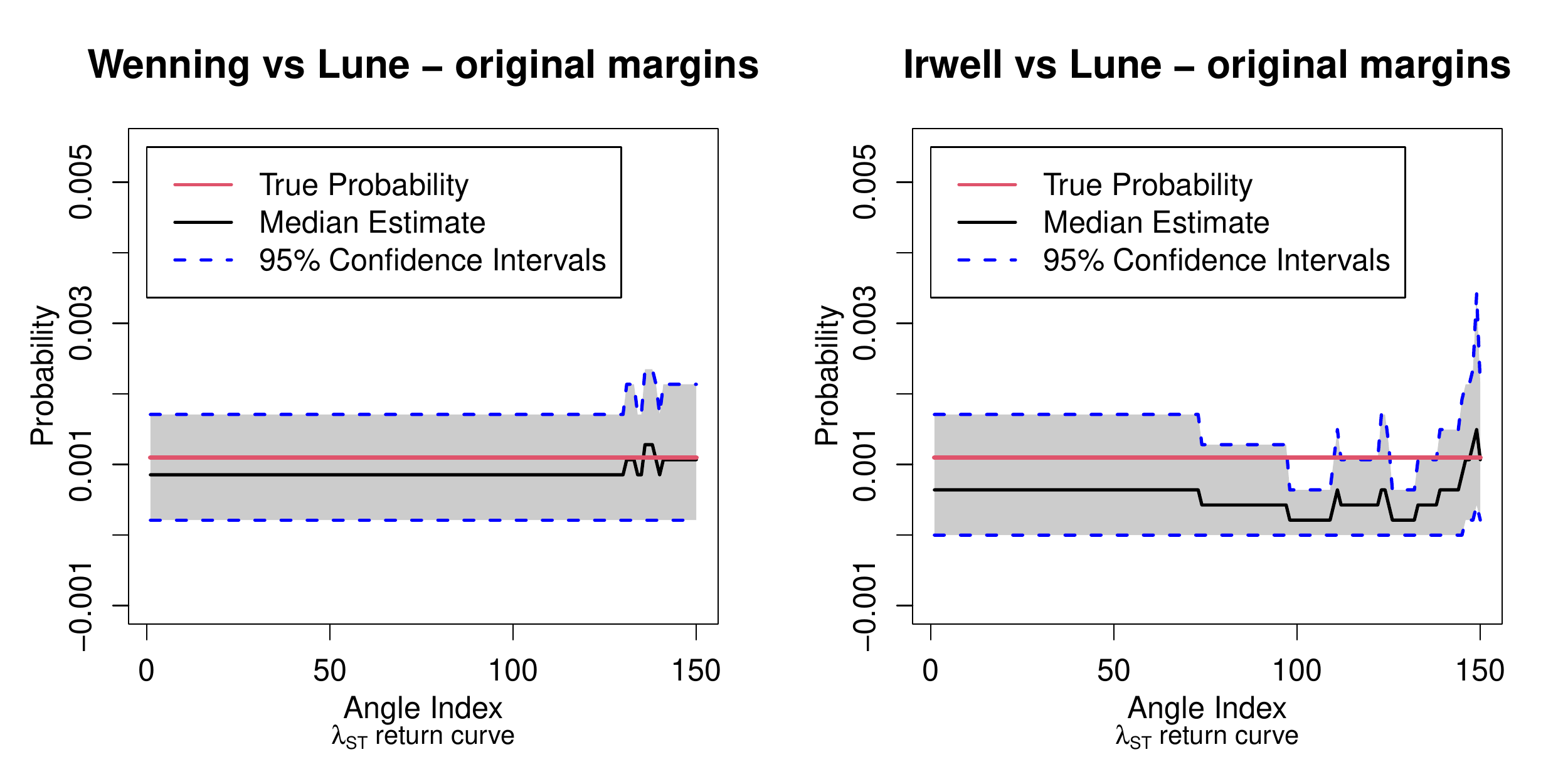}
    \caption{Diagnostic plots of the return curve estimates from the $\hat{\lambda}_{ST}$ estimator for the first and fifth gauge site pairings. The black and red lines indicate the empirical median and true survival probabilities, respectively, with 95\% bootstrapped confidence intervals denoted by the shaded regions.}
    \label{fig:diag_st}
\end{figure}

\begin{figure}[!h]
    \centering
    \includegraphics[width=\textwidth]{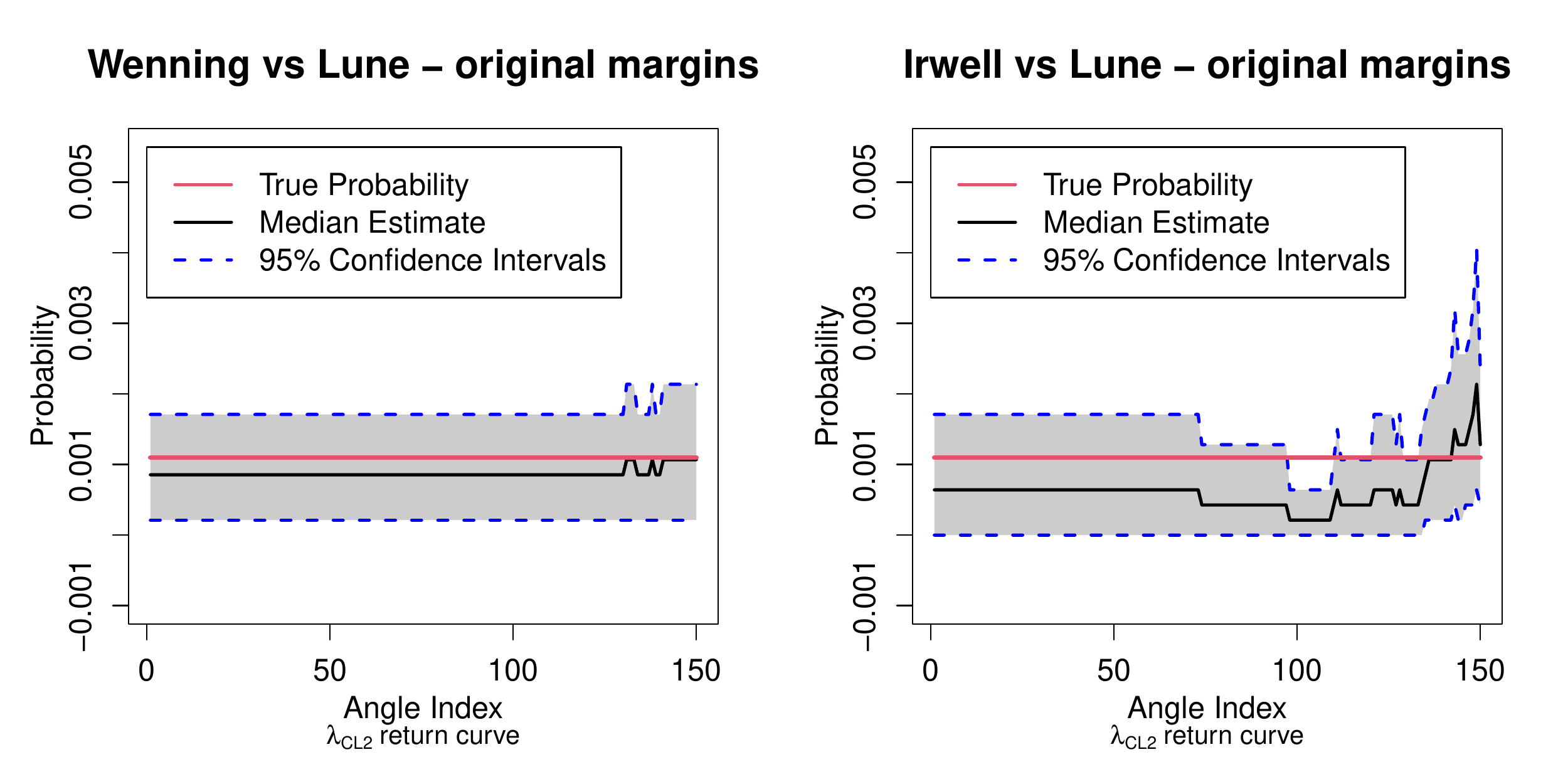}
    \caption{Diagnostic plots of the return curve estimates from the $\hat{\lambda}_{CL2}$ estimator for the first and fifth examples. The black and red lines indicate the empirical median and true survival probabilities, respectively, with 95\% bootstrapped confidence intervals denoted by the shaded regions}
    \label{fig:diag_cl2}
\end{figure}

Finally, Figures \ref{fig:uncert_rc2} and \ref{fig:uncert_rc} illustrate estimated return curve uncertainty intervals obtained using the $\hat{\lambda}_{ST}$ and $\hat{\lambda}_{CL2}$ estimators, respectively, for the first and fifth gauge site pairings. In both figures, one can observe the contrast in shapes of the uncertainty regions between the two site pairings. 

\begin{figure}[!h]
    \centering
    \includegraphics[width=\textwidth]{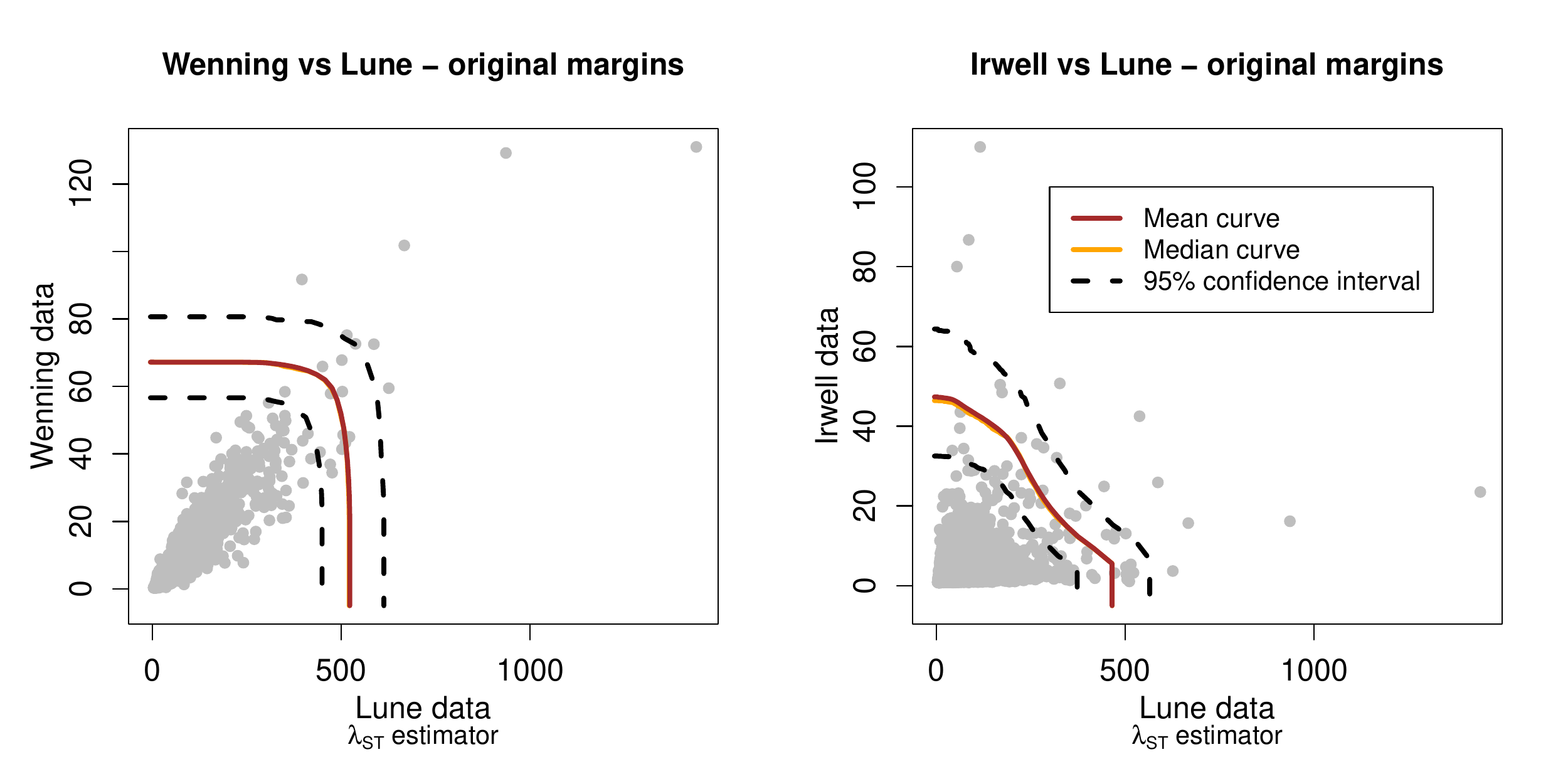}
    \caption{Median and mean return curve estimates in orange and brown, respectively, obtained using the $\hat{\lambda}_{ST}$ estimator for the first and fifth examples. The black dotted lines indicate 95\% confidence intervals.}
    \label{fig:uncert_rc2}
\end{figure}

\begin{figure}[!h]
    \centering
    \includegraphics[width=\textwidth]{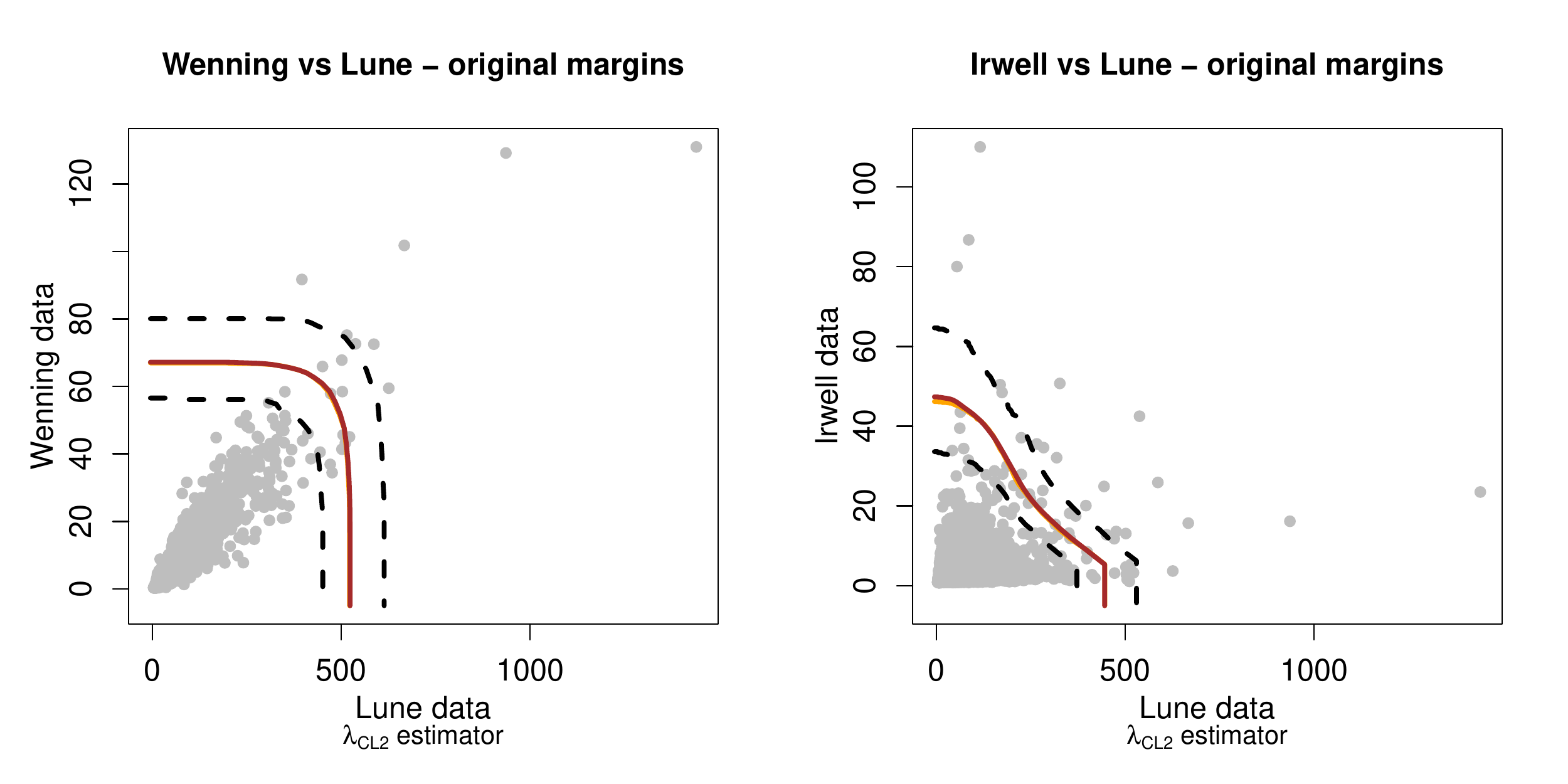}
    \caption{Median and mean return curve estimates in orange and brown, respectively, obtained using the $\hat{\lambda}_{CL2}$ estimator for the first and fifth examples. The black dotted lines indicate 95\% confidence intervals.}
    \label{fig:uncert_rc}
\end{figure}

\clearpage

\bibliography{library,nrfa}

\end{document}